\documentclass[11pt]{article}

\usepackage{amsmath}
\usepackage{amsthm}
\usepackage{amssymb}
\usepackage{rotating}
\usepackage{array}
\usepackage{multirow}
\usepackage{makecell}
\usepackage[table]{xcolor}
\usepackage{enumitem}
\usepackage{mathtools}
\usepackage{setspace}
\usepackage{fullpage}
\usepackage{hyperref}
\usepackage{xspace}
\usepackage{fullpage}
\usepackage{authblk}

\usepackage{algorithm}
\usepackage[noend]{algpseudocode}
\algnewcommand\Input{\item[\textbf{Input:}]}
\algnewcommand\Output{\item[\textbf{Output:}]}

\newtheorem{definition}{Definition}
\newtheorem{theorem}{Theorem}

\DeclareMathOperator* {\argmin} {arg\,min}
\DeclareMathOperator* {\select} {select}

\newcommand{\R}{\mathbb{R}}
\newcommand{\N}{\mathbb{N}}
\newcommand{\G}{\mathbb{G}}

\newcommand{\CC}{K}
\newcommand{\Nb}{K}
\newcommand{\TT}{\langle T \rangle}
\newcommand{\lt}{\lambda}
\newcommand{\al}{\lt^*}
\newcommand{\ps}{\Pi}
\newcommand{\cdegr}{c_D}
\newcommand{\cclos}{c_C}
\newcommand{\cbetw}{c_B}

\newcommand{\vs}{v^\dagger}
\newcommand{\FA}{\widehat{A}}
\newcommand{\FR}{\widehat{R}}
\newcommand{\Add}{A^*}
\newcommand{\Rem}{R^*}
\newcommand{\thr}{\omega}

\title{Hiding in Temporal Networks}

\author[a]{Marcin Waniek}
\author[b]{Petter Holme}
\author[a]{Talal Rahwan}

\affil[a]{New York University Abu Dhabi, Abu Dhabi, UAE}
\affil[b]{Tokyo Institute of Technology, Tokyo, Japan}

\date{}

\begin{document}

\maketitle

\begin{abstract}
Social network analysis tools can infer various attributes just by scrutinizing one's connections. Several researchers have studied the problem faced by an evader whose goal is to strategically rewire their social connections in order to mislead such tools, thereby concealing their private attributes. However, to date this literature has only considered static networks, while neglecting the more general case of temporal networks, where the structure evolves over time. Driven by this observation, we study how the evader can conceal their importance from an adversary armed with temporal centrality measures. We consider computational and structural aspects of this problem: Is it computationally feasible to calculate optimal ways of hiding? If it is, what network characteristics facilitate hiding? This topic has been studied in static networks, but in this work, we add realism to the problem by considering temporal networks of edges changing in time. We find that it is usually computationally infeasible to find the optimal way of hiding. On the other hand, by manipulating one's contacts, one could add a surprising amount of privacy. Compared to static networks, temporal networks offer more strategies for this type of manipulation and are thus, to some extent, easier to hide in.
\end{abstract}

\section{Introduction}\label{sec:introduction}

% Privacy threats
The increasing sophistication and ubiquity of computer-based invigilation tools is a persistent threat to the  privacy of the general public.
The increasing number of privacy-related scandals, such as Cambridge Analytica using data of millions of Facebook users for political agendas~\cite{isaak2018user}, demonstrates just how vulnerable our private information is in the age of social media. Network data, from social media in particular, can be used to uncover sensitive information, such as sexual orientation~\cite{jernigan2009gaydar}, relationship status~\cite{gross2005information}, or political views~\cite{mislove2010you}.

% Protecting privacy in networks
Due to this vulnerability of network data, a number of privacy-preservation solutions have been proposed, both in terms of legislature~\cite{EU:2016:gdpr} and algorithmic solutions~\cite{Lane:et:al:2014,Kearns:et:al:2016}.
Most of the literature puts the responsibility of protecting the system users' privacy in the hands of a centralized authority~\cite{hay2007anonymizing,khatri2010designing,fung2010privacy}, but such an authority might be prone to error and negligence.
A different body of literature proposes methods that can be applied by members of the social network to protect their own privacy, without having to rely on any central entities~\cite{yu2018target,michalak2017}.

% Temporal networks are increasingly popular
Nevertheless, so far, the development of privacy-protection schemes for network data has focused on networks that are static, i.e., whose structure do not change over time. At the same time, network science researchers are starting to shift their attention to temporal networks---the more general case where the structure is allowed to change~\cite{holme2012temporal}. In many domains that involve dynamic changes, momentary contacts, and processes unfolding over time, the added complexity of a temporal-network approach can be justified by an added predictive and explanatory power~\cite{holme2012temporal}. Temporal networks already found application in such varied areas as communication~\cite{eckmann2004entropy,candia2008uncovering}, microbiology~\cite{przytycka2010toward,rao2007inferring}, and face-to-face interactions~\cite{takaguchi2011predictability}.
Nevertheless, so far no privacy-protection solutions have been proposed for temporal networks.

% What are we doing in this work
Motivated by the relevance of the growing privacy threats and the increasing popularity of temporal networks, in this work we examine how a member of a temporal network can avoid being detected by temporal centrality measures. We set out to investigate the issue from both theoretical and empirical standpoints. As for the theoretical analysis, we evaluate the computational complexity of the problem faced by an evader who wishes to rewire the network in order to obscure their central position in it. We consider both the decision version of the problem---if it is possible to find an optimal solution in polynomial time---as well as the optimization version---if it is possible to identify a solution that is guaranteed to be within a certain bound from optimum. As for the empirical analysis, we consider several real-life temporal network datasets and investigate how effectively the top nodes in the temporal centrality rankings can conceal their importance. We not only observe the results of the hiding process but also identify, using regression analysis, the properties of the nodes that allow them to conceal their importance effectively. Altogether, our study is the first analysis of strategic evasion of social network analysis tools in temporal networks.

\section{Related Work}
\label{sec:related-work}

% Temporal networks
In recent years there has been a growing interest in temporal networks, where the network structure is not static, but instead changes with the passage of time~\cite{holme2012temporal}.
Temporal networks found particularly relevant applications in epidemics, where they have been used to predict the infection's reproduction number~\cite{holme2015basic} and to construct static graphs based on temporal contact data~\cite{holme2013epidemiologically}.
In the context of our work, an essential class of tools for the analysis of temporal networks are centrality measures~\cite{kim2012temporal}.
The approach to their design varies greatly, ranging from the analysis of network flows~\cite{tang2010analysing} and shortest temporal paths~\cite{pan2011path}, to the applications of eigenvector-like techniques~\cite{taylor2017eigenvector,lv2019eigenvector,taylor2019supracentrality}.

% Hiding in networks
The topic of avoiding detection by social network analysis tools recently received some attention in the literature devoted to static networks.
The greatest amount of attention is focused on hiding from centrality measures, either by lowering the node's centrality, either in absolute terms~\cite{waniek2018hiding} or in relative terms~\cite{waniek2017construction,dey2019covert}.
Other works provide an axiomatic characterization of centrality measures that are resilient to being fooled~\cite{was2020manipulability}, analyze the possible strategies of an adversary who is aware of the existence of nodes that want to hide themselves~\cite{waniek2021strategic}, or consider the problem of hiding from centrality measures in multilayer networks~\cite{waniek2020hiding}. Other hiding problems considered in the literature include preventing the identification of closely-cooperating groups of nodes by community detection algorithms~\cite{waniek2018hiding}, avoiding the detection of private relationships by link prediction algorithms~\cite{waniek2019hide,zhou2019attacking}, and investigating the possibility of concealing the source of network diffusion from source detection algorithms~\cite{waniek2021social}. All of these works consider only static networks.

\section{Preliminaries}
\label{sec:preliminaries}

\subsection{Temporal Networks}

Throughout the article, we will let $\TT$ denote a time interval of $T$ discrete time steps, i.e., $\TT = \{0, \ldots, T-1\}$.
We will sometimes refer to a particular $t \in \TT$ as the \textit{moment} $t$.
Let us denote by $G = (V,\CC,T) \in \G$ a temporal network, where $V$ is the set of $n$ nodes, $\CC \subseteq V \times V \times \TT$ is the set of contacts, and $T$ is the duration of the time interval during which the contacts in $\CC$ take place.
We denote a \textit{contact} (sometimes also called a \textit{temporal edge}) between nodes $v$ and $w$ at time $t$ by $(v,w,t)$.
In this work we only consider \textit{undirected} temporal networks, i.e., we do not discern between contacts $(v,w,t)$ and $(w,v,t)$.
Moreover, we assume that networks do not contain self-contacts, i.e., $\forall_{v \in V} \forall_{t \in \TT}(v,v,t) \notin \CC$.
We denote all contacts of a given node $v$ by $\Nb_G(v)$.
Finally, for $\CC' \subseteq V \times V \times \TT$ we denote by $G \cup \CC'$ the effect of adding the set of contacts $\CC'$ to $G$, i.e., $G \cup \CC' = (V, \CC \cup \CC',T)$.

A \textit{time-respecting path} in a temporal network $G = (V,\CC,T)$ is an ordered sequence of distinct contacts from $\CC$, $\langle c_1, \ldots, c_k \rangle$, in which for every two consecutive contacts $(v,w,t)$ and $(v',w',t')$ we have that $w=v'$ and $t'>t$.
The \textit{duration} of the path is the time difference between the first and the last contact in the path, i.e., the duration of a path $\langle (v,w,t), \ldots, (v',w',t') \rangle$ is $t'-t$.
Let $\ps_G(v,w)$ denote the set of temporal paths from $v$ to $w$ with the minimal duration.
We say that we can \textit{reach} node $w$ from node $v$ at time $t$ if there exists a time-respecting path from $v$ to $w$ that starts at time greater than or equal to $t$.
The \textit{latency} between $v$ and $w$ at time $t$ is the shortest time it takes to reach $w$ from $v$ starting at time $t$ along time-respecting paths, we denote it by $\lt_G(v,w,t)$.
If no such time-respecting path exists then $\lt_G(v,w,t)=\infty$.
We denote \textit{average latency} between $v$ and $w$ by $\al(v,w)$.
Notice that average latency is the area under the latency plot, divided by the length of the time interval $T-1$.
Following Pan and Saram{\"a}ki~\cite{pan2011path} we assume a pair-specific temporal boundary condition, where for every pair of nodes, the first time-respecting path between them is repeated after the end of the observed time interval when computing the average latency for that pair.
Hence, if there exists at least one time-respecting path between the pair of nodes, then the average latency between them is finite.
An example of a latency plot used to compute the average latency is presented in Figure~\ref{fig:latency-example}.

\begin{figure}[t]
\centering
\includegraphics[width=.6\linewidth]{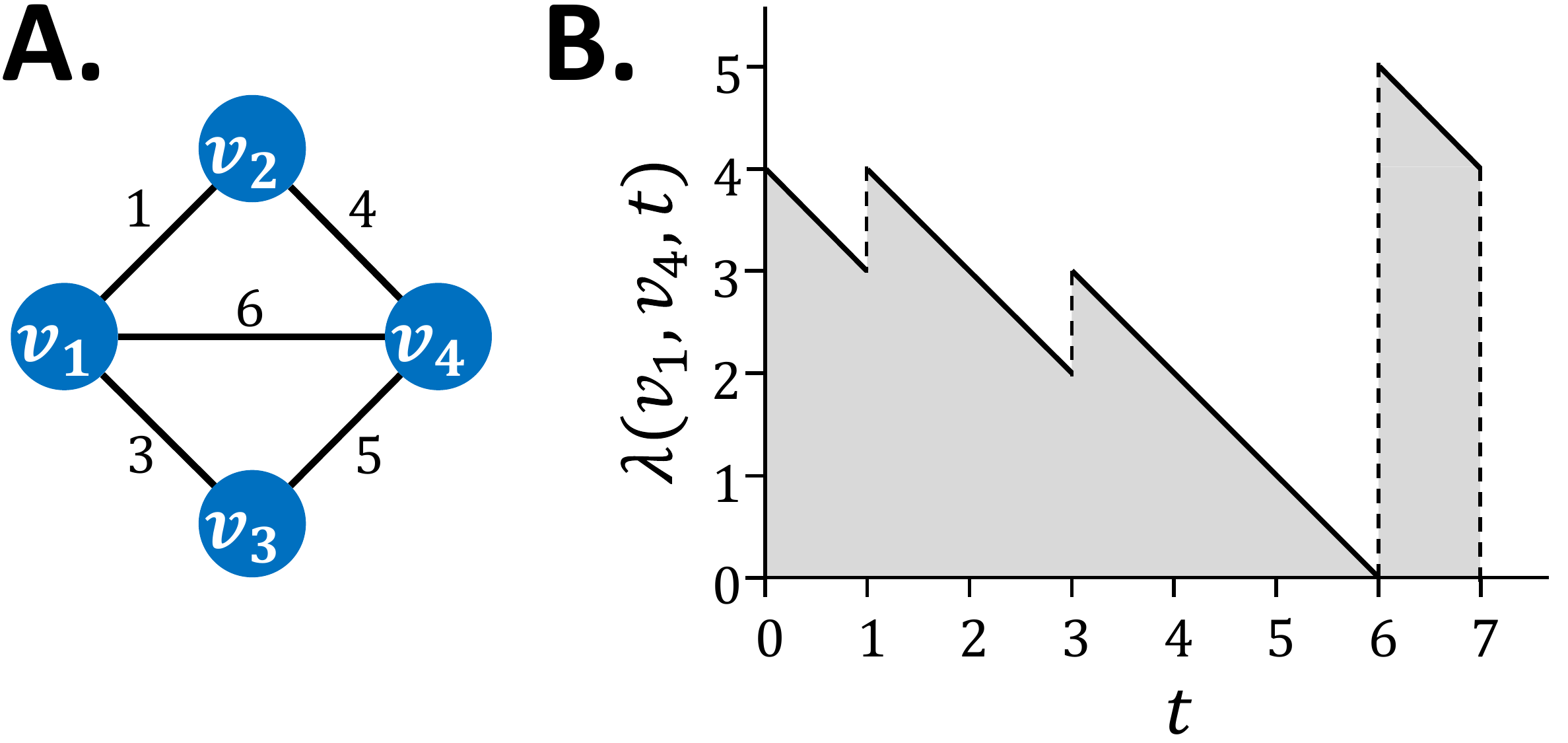}
\caption{
\textbf{An example of temporal network and a latency plot.}
\textbf{A.} A temporal network with times of each contact denoted as number next to edges.
\textbf{B.} A plot of latency between $v_1$ and $v_4$ for under assumption that the duration of time interval is $T=7$.
The gray area under plot is proportional to average latency between $v_1$ and $v_4$.
Notice that the latency in time interval $(6,7]$ is finite due to the assumption about a pair-specific temporal boundary condition~\cite{pan2011path}.
}
\label{fig:latency-example}
\end{figure}

To make the notation more readable, we will often omit the temporal network itself from the notation whenever it is clear from the context, e.g., by writing $\lt(v,w,t)$ instead of $\lt_G(v,w,t)$.
This applies not only to the notation presented thus far, but rather to all notation in this article.

\subsection{Temporal Centrality Measures}

Centrality measures quantify the importance of a given node in a network.
The concept has also been extended to temporal networks~\cite{kim2012temporal}.
In this work we consider the following temporal centrality measures:

\begin{itemize}

\item Degree temporal centrality~\cite{kim2012temporal}---importance of a node $v$ corresponds to its number of contacts.
\[
\cdegr\!\left(\left(V,\!\CC,\!T\right)\!,\!v\right) = \frac{\sum_{t \in \TT}|\{w\!\in\!V\!:\! (v,w,t)\!\in\!\CC\}|}{(n-1)T}.
\]

\item Closeness temporal centrality~\cite{pan2011path}---importance of a node $v$ corresponds to its average latency to other nodes.
\[
\cclos\!\left(\left(V,\!\CC,\!T\right)\!,\!v\right) = \frac{1}{n-1} \sum_{w \in V : v \neq w}\frac{1}{\al(v,w)}.
\]

\item Betweenneess temporal centrality~\cite{tang2010analysing}---importance of a node $v$ corresponds to the percentage of shortest temporal paths between pairs of other nodes passing through $v$.
\begin{align*}
\cbetw\!\left(\left(V,\!\CC,\!T\right)\!,\!v\right) =& \frac{1}{(n-1)(n-2)T} \\
& \sum_{\substack{u,w \in V \setminus \{v\}:\\ u \neq w}} \sum_{t=0}^{T-1} \frac{|\widetilde{\ps}^v_t(u,\!w)|}{|\ps(u,\!w)|}
\end{align*}
where $\widetilde{\ps}^v_t(u,w)$ is the set of shortest temporal paths from $u$ to $w$ such that $v$ belongs to the path and either contacts between $v$ and its predecessor and successor take place at moment $t$, or contact between $v$ and its predecessor takes place at or before moment $t$, while the contact between $v$ and its successor takes place after moment $t$:
\begin{align*}
\widetilde{\ps}^v_t(u,\!w) =& \{\pi\!\in\!\ps(u,\!w) : (v'\!\!,\!v,\!t')\!\in\!\pi \land (v,\!v''\!\!,\!t'')\!\in\!\pi \\
& \land (t' = t = t'' \lor t' \leq t < t'')\}.
\end{align*}

\item Eigenvector temporal centrality~\cite{lv2019eigenvector}---importance of a node $v$ corresponds to the importance of its neighbors. The temporal version of the eigenvector centrality was defined in a number of ways~\cite{taylor2017eigenvector,huang2017centrality,yin2018inter,taylor2019tunable}.
In our work we use algorithm by Lv et al.~\cite{lv2019eigenvector}, as it allows to efficiently process even relatively large networks.

\end{itemize}

\section{Theoretical Analysis}
\label{sec:theoretical-analysis}

We now present the formal definition of the computational problems faced by the evader, starting with finding an optimal way of hiding within a certain budget.

\begin{definition}[Temporal Hiding]
An instance of the Temporal Hiding problem is defined by a tuple~$(G,\vs,\FA,\FR,b,c,\thr)$, where $G= (V,\CC,T)$ is a temporal network, $\vs$ is the evader, $\FA \subseteq V \times V \times \TT$ is the set of contacts allowed to be added, $\FR \subseteq \CC$ is the set of contacts allowed to be removed, $b \in \N$ is a budget specifying the maximum number of contacts that can be added or removed, $c \colon \G \times V \rightarrow \R$ is a temporal centrality measure, and $\thr \in \N$ is a chosen safety threshold.
The goal is then to identify two sets, $\Add \subseteq \FA$ and $\Rem \subseteq \FR$, such that $|\Add|+|\Rem| \leq b$ and:
\[
\left|\left\{ w \in V : c(G^*,w) > c(G^*,\vs) \right\}\right| \geq \thr
\]
where $G^* = \left(V, (\CC \cup \Add) \setminus \Rem, T\right)$.
\end{definition}

We also present an approximation version of the problem, where the goal of the evader is to satisfy a certain safety threshold using as few network modifications as possible.

\begin{definition}[Minimum Temporal Hiding]
An instance of the Minimum Temporal Hiding problem is defined by a tuple $(G,\vs,\FA,\FR,c,\thr)$, where $G= (V,\CC,T)$ is a temporal network, $\vs$ is the evader, $\FA \subseteq V \times V \times \TT$ is the set of contacts allowed to be added, $\FR \subseteq \CC$ is the set of contacts allowed to be removed, $c \colon \G \times V \rightarrow \R$ is a temporal centrality measure, and $\thr \in \N$ is a chosen safety threshold.
The goal is then to identify two sets, $\Add \subseteq \FA$ and $\Rem \subseteq \FR$, such that the sum of their sizes $|\Add|+|\Rem|$ is as small as possible and and:
\[
\left|\left\{ w \in V : c(G^*,w) > c(G^*,\vs) \right\}\right| \geq \thr
\]
where $G^* = \left(V, (\CC \cup \Add) \setminus \Rem, T\right)$.
\end{definition}

\begin{table*}[t!]
\caption{Summary of our computational complexity results.}
\label{tab:decision-summary}
\centering
\small
\begin{tabular}{lcc}
\hline
Centrality & Temporal Hiding & Minimum Temporal Hiding \\
\hline
Degree & NP-complete (Theorem~\ref{thrm:npc-degree}) & We show a $2$-approximation algorithm (Theorem~\ref{thrm:appr-degree}) \\
Closeness & NP-complete (Theorem~\ref{thrm:npc-closeness}) & Inapproximable within $(1-\epsilon) \ln |\FA|$ for any $\epsilon > 0$ (Theorem~\ref{thrm:appr-closeness}) \\
Betweenness & NP-complete (Theorem~\ref{thrm:npc-betweenness}) & Inapproximable within $(1-\epsilon) \ln |\FA|$ for any $\epsilon > 0$ (Theorem~\ref{thrm:appr-betweenness}) \\
\hline
\end{tabular}
\end{table*}

Table~\ref{tab:decision-summary} presents the summary of our computational complexity results.

\begin{theorem}
\label{thrm:npc-degree}
Temporal Hiding problem is NP-complete given the degree temporal centrality.
\end{theorem}

\begin{proof}
The problem is trivially in NP, since after adding and removing a given set of contacts we can computed the ranking of all nodes according to the degree temporal centrality in polynomial time.

We will now prove the NP-hardness of the problem.
To this end, we will show a reduction from the NP-complete \textit{Finding $k$-Clique} problem.
The decision version of this problem is defined by a simple network, $H=(V,E)$, where $V=\{v_1,\ldots,v_n\}$, and a constant $k \in \N$, where the goal is to determine whether there exist $k$ nodes forming a clique in $H$.

Let $(H,k)$ be a given instance of the Finding $k$-Clique problem.
Let us assume that $n=2$, i.e., network $H$ has at least two nodes.

We will now construct an instance of the Temporal Hiding problem. First, let us construct a temporal network $G=(V',\CC,T)$ where:
\begin{itemize}
\item $V' = V \cup \{ \vs \} \cup \bigcup_{i=1}^{n} \bigcup_{j=1}^{n-k+1} \{u^i_j\}$,
\item $\CC\!=\!\bigcup_{v_i \in V} \bigcup_{u^i_j \in V} \{(v_i,\!u^i_j,\!0)\} \cup \bigcup_{v_i \in V} \{(\vs\!,\!v_i,\!0)\}$,
\item $T=1$.
\end{itemize}
An example of the construction of the network $G$ is presented in Figure~\ref{fig:npc-degree}.

Now, consider the instance $(G,\vs,\FA,\FR,b,c,\thr)$ of the Temporal Hiding problem, where:
\begin{itemize}
\item $G$ is the temporal network we just constructed,
\item $\vs$ is the evader,
\item $\FA = \{(v_i,v_j,0) : (v_i,v_j) \in E\}$, i.e., only edges existing in $H$ can be added to $G$, all of them in moment $0$,
\item $\FR = \emptyset$, i.e., none of the edges can be removed,
\item $b=\frac{k(k-1)}{2}$,
\item $c$ is the temporal degree centrality,
\item $\thr=k$ is the safety threshold.
\end{itemize}

Since $\FR = \emptyset$, for any solution to the constructed instance of the Temporal Hiding problem we must have $\Rem = \emptyset$.
Hence, we will omit mentions of $\Rem$ in the remainder of the proof, and we will assume that a solution consists just of $\Add$.

Note that the number of contacts of the evader $\vs$ in $G$ is $n$, and it does not change after the addition of any $A \subseteq \Add$ (as we can only add edges between the members of $V$).
Furthermore, note that the number of contacts of every node $u^i_j$ is $1$, and it cannot be increased.
Therefore, the only nodes that can contribute to satisfying the safety threshold (by increasing their centrality to a value greater than that of the evader) are the nodes in $V$.
The number of contacts of any $v_i$ in $G$ is $n-k+2$ (as it contacts with $n-k+1$ nodes $u^i_j$ and with the node $\vs$).
Therefore, for a given $v_i$ to have a greater number of contacts than $\vs$, we have to add to $G$ at least $k-1$ contacts incident with $v_i$.

\begin{figure}[t]
\centering
\includegraphics[width=.5\linewidth]{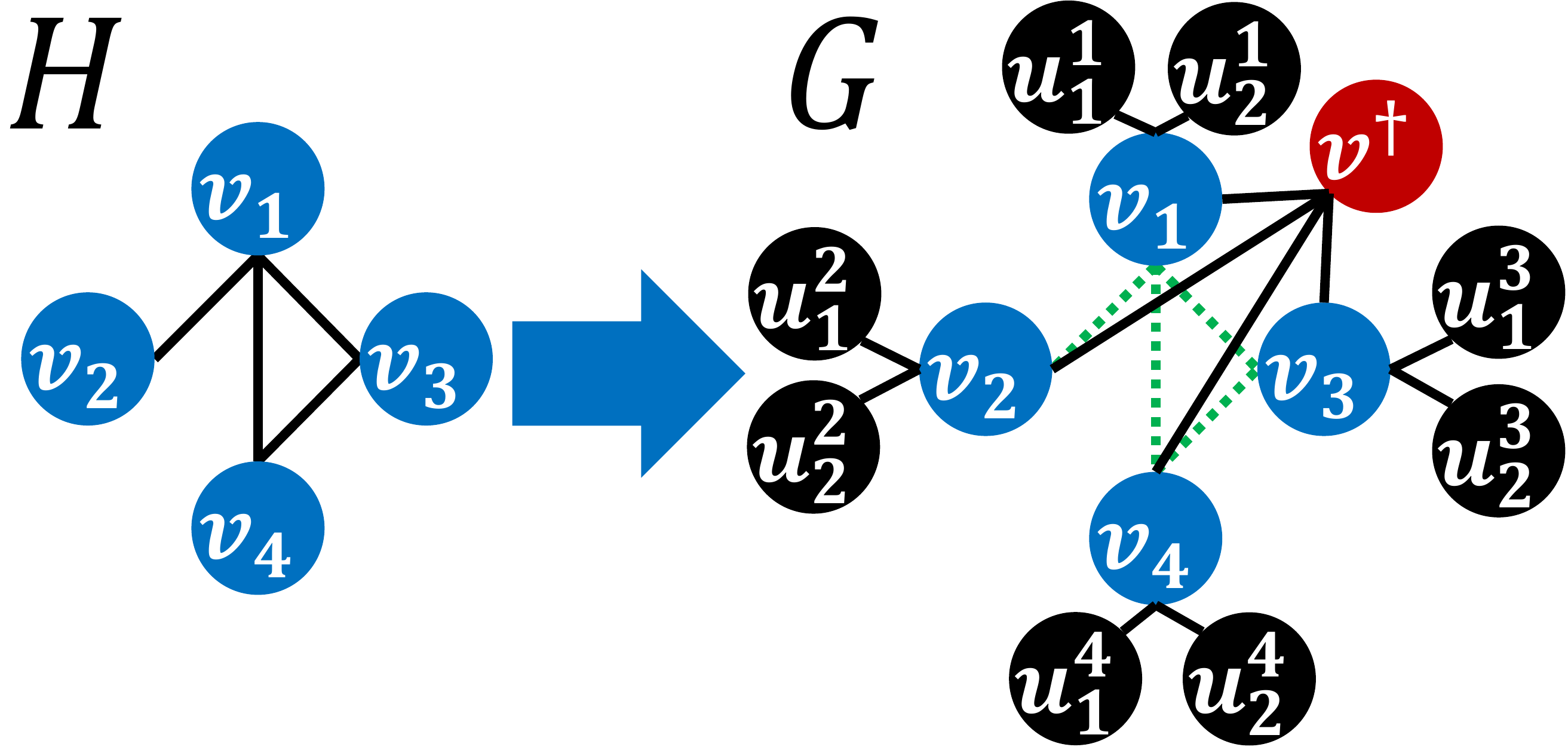}
\caption{
Construction used in the proof of Theorem~\ref{thrm:npc-degree}.
The network $H$ is given as part of the Finding $k$-Clique problem instance, while the network $G$ is constructed as part of the Temporal Hiding problem instance.
The blue nodes in $G$ are those corresponding to the nodes in $H$, while the evader is marked red.
All contacts occur at moment $0$.
Green dotted lines depict the contacts allowed to be added.
}
\label{fig:npc-degree}
\end{figure}

We will now show that the constructed instance of the Temporal Hiding problem has a solution if and only if the given instance of the Finding $k$-Clique problem has a solution.

Assume that there exists a solution to the given instance of the Finding $k$-Clique problem, i.e., a subset $V^* \subseteq V$ forming a $k$-clique in $H$.
We will show that $\Add = \{(v_i,v_j,0): v_i,v_j \in V^*\}$ is a solution to the constructed instance of the Temporal Hiding problem.
First, notice that indeed $\Add \subseteq \FA$, as $\FA$ contains all edges from $H$ at moment $0$, and $V^* \times V^*$ is a clique in $H$.
Note also that adding $\Add$ to $G$ increases the number of contacts of $k$ nodes in $V^*$ to $n+1$, i.e., to a value greater than the number of contacts of the evader.
We showed that if there exists a solution to the given instance of the Finding $k$-Clique problem, then there also exists a solution to the constructed instance of the Temporal Hiding problem.

Assume that there exists a solution $\Add$ to the constructed instance of the Temporal Hiding problem.
We will show that $V^* = \bigcup_{(v_i,v_j,0) \in \Add} \{v_i,v_j\}$ forms a $k$-clique in $H$.
Notice that since $\Add$ is a solution, it increases the number of contacts of at least $k$ nodes in $V$ (since the safety threshold is $\thr=k$) by at least $k-1$.
However, since the budget is $b=\frac{k(k-1)}{2}$, adding $\Add$ must increase the number of contacts of \textit{exactly} $k$ nodes in $V$ by \textit{exactly} $k-1$.
If such a choice is available, the nodes in $V^*$ form a clique in $\FA$ at moment $0$, therefore, they also form a clique in $H$.
We showed that if there exists a solution to the constructed instance of the Temporal Hiding problem, then there also exists a solution to the given instance of the Finding $k$-Clique problem.
This concludes the proof.
\end{proof}

\begin{theorem}
\label{thrm:npc-closeness}
Temporal Hiding problem is NP-complete given the closeness temporal centrality.
\end{theorem}

\begin{proof}
The problem is trivially in NP since after adding and removing a given set of contacts, we can compute the ranking of all nodes according to the closeness temporal centrality in polynomial time.

We will now prove that the problem is NP-hard.
To this end, we will show a reduction from the NP-complete \textit{Set Cover} problem.
The decision version of this problem is defined by a universe, $U=\{u_1,\ldots,u_{|U|}\}$, a collection of sets $S=\{S_1,\ldots,S_{|S|}\}$ such that $\forall_i S_i \subset U$, and an integer $k \in \N$ where the goal is to determine whether there exist $k$ elements of $S$ the union of which equals $U$.

Let $(U,S,k)$ be a given instance of the Set Cover problem.
Let us assume that $|S| \geq 3$, note that all instances where $|S| < 3$ can be solved in polynomial time.
We will now construct an instance of the Temporal Hiding problem.
First, let us construct a temporal network $G=(V,\CC,T)$ where:
\begin{itemize}
\item $V = \{ \vs, z \} \cup U \cup S \cup \bigcup_{i=1}^{|S|-k} \{y_i\} \cup \bigcup_{i=1}^{|U|+|S|-1} \{x_i\}$,
\item $\CC = \bigcup_{x_i \in V}\{(\vs,x_i,0)\} \cup \bigcup_{y_i \in V}\{(z,y_i,0)\} \cup \bigcup_{u_i \in S_j}\{(u_i,S_j,1)\}$,
\item $T=2$.
\end{itemize}
An example of the construction of the network $G$ is presented in Figure~\ref{fig:npc-closeness}.

\begin{figure}[t]
\centering
\includegraphics[width=.6\linewidth]{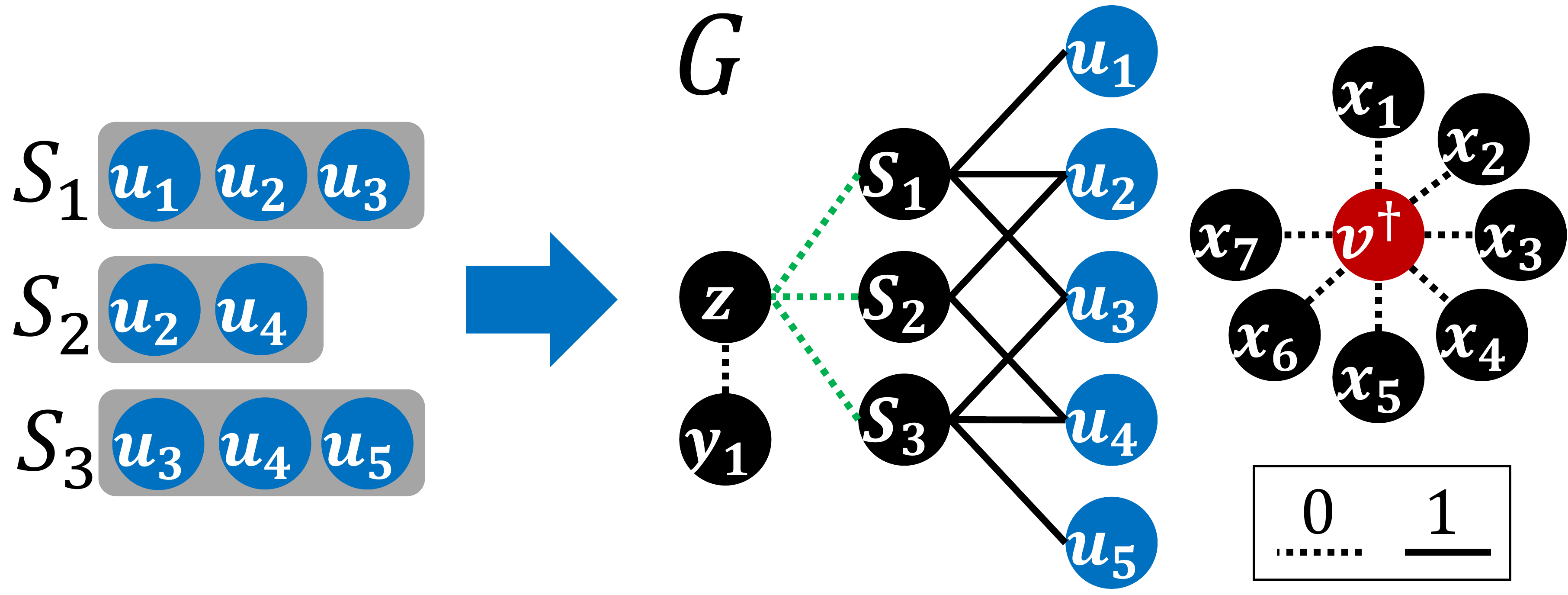}
\caption{
Construction used in the proof of Theorem~\ref{thrm:npc-closeness}.
On the left there is an instance of the Set Cover problem, while on the right there is a network $G$ constructed as part of the Temporal Hiding problem instance.
The blue nodes in $G$ correspond to the elements of the universe $U$, while the evader is marked red.
Dotted lines correspond to contacts at moment $0$, and solid lines to contact at moment $1$.
Green lines depict the contacts allowed to be added.
}
\label{fig:npc-closeness}
\end{figure}

Now, consider the instance $(G,\vs,\FA,\FR,b,c,\thr)$ of the Temporal Hiding problem, where:
\begin{itemize}
\item $G$ is the temporal network we just constructed,
\item $\vs$ is the evader,
\item $\FA = \{(z,S_i,0) : S_i \in V\}$,
\item $\FR = \emptyset$, i.e., none of the edges can be removed,
\item $b=k$,
\item $c$ is the temporal closeness centrality,
\item $\thr=1$ is the safety threshold.
\end{itemize}

Since $\FR = \emptyset$, for any solution to the constructed instance of the Temporal Hiding problem we must have $\Rem = \emptyset$.
Hence, we will omit mentions of $\Rem$ in the remainder of the proof, and we will assume that a solution consists just of $\Add$.

First, let us analyze the temporal closeness centrality values of the nodes in $G$ after the addition of any $A \subseteq \FA$.
Let $q_A$ denote the number of nodes $u_j \in U$ such that $\al(z,u_j)<\infty$, i.e., the number of nodes $u_j \in U$ such that $z$ is connected (via the contacts from $A$) with at least one node $S_i$ connected with $u_j$ (notice that there are no other possibilities for $z$ to have finite average latency to a node in $U$).
The temporal closeness centrality values can be computed as follows:
\begin{itemize}
\item $\cclos(G \cup A, \vs) = \frac{2(|U|+|S|-1)}{n-1}$,
\item $\cclos(G \cup A, z) = \frac{2(|S|-k+|A|+q_A)}{n-1}$,
\item $\cclos(G \cup A, S_i) \leq \frac{2(|U|+1)}{n-1} < \cclos(G \cup A, \vs)$ for every $S_i \in V$ such that $(z,S_i,0) \in A$,
\item $\cclos(G \cup A, S_i) \leq \frac{2|U|}{n-1} < \cclos(G \cup A, \vs)$ for every $S_i \in V$ such that $(z,S_i,0) \notin A$,
\item $\cclos(G \cup A, u_i) = \frac{2(|\{S_j \in S: u_i \in S_j\}|)}{n-1} \leq \frac{2|S|}{n-1} < \cclos(G \cup A, \vs)$ for every $u_i \in V$,
\item $\cclos(G \cup A, x_i) = \frac{2}{n-1} < \cclos(G \cup A, \vs)$ for every $x_i \in V$,
\item $\cclos(G \cup A, y_i) = \frac{2}{n-1} < \cclos(G \cup A, \vs)$ for every $y_i \in V$.
\end{itemize}

Since the safety threshold is $\thr=1$, only one node has to have greater temporal closeness centrality than $\vs$ after he addition of a given $A \subseteq \FA$ in order for said $A$ to be solution to the constructed instance of the Temporal Hiding problem.
However, notice that the only node that can have greater temporal closeness centrality than $\vs$ (and therefore higher position in the ranking) is $z$.
In other words, a given $A \subseteq \FA$ is a solution to the constructed instance of the Temporal Hiding problem if and only if we have $\cclos(G \cup A, z) > \cclos(G \cup A, \vs)$.

We will now show that the constructed instance of the Temporal Hiding problem has a solution if and only if the given instance of the Set Cover problem has a solution.

Assume that there exists a solution to the given instance of the Set Cover problem, i.e., a subset $S^* \subseteq S$ of size $k$ the union of which is the universe $U$.
We will show that $\Add = \{z\} \times S^* \times \{0\}$ (i.e., connecting $z$ with nodes corresponding to all sets in $S^*$) is a solution to the constructed instance of the Temporal Hiding problem.
First, notice that after the addition of $\Add$ there exists a time-respecting path from $z$ to every node $u_j \in U$, leading through the node corresponding to an element $S_i \in S^*$ containing $u_j$, with which $z$ is now connected.
Hence, we have that $q_{\Add} = |U|$ and $|\Add|=k$, which gives us:
\begin{align*}
\cclos(G \cup \Add, z) = \frac{2(|S|+|U|)}{n-1} >& \frac{2(|U|+|S|-1)}{n-1} \\
& = \cclos(G \cup \Add, \vs).
\end{align*}
Therefore, after the addition of $\Add$ node $z$ has greater temporal closeness centrality than the evader.
We showed that if there exists a solution to the given instance of the Set Cover problem, then there also exists a solution to the constructed instance of the Temporal Hiding problem.

Assume that there exists a solution $\Add$ to the constructed instance of the Temporal Hiding problem.
We will show that $S^* = \{S_i \in S : (z,S_i,0) \in \Add\}$ covers the universe $U$.
Let us compute the difference between $\cclos(G \cup \Add, \vs)$ and $\cclos(G \cup \Add, z)$:
\begin{align*}
\cclos(G \cup \Add, z)-\cclos(G&\cup \Add, \vs) \\
&= \frac{2(|\Add|-k+q_{\Add}-|U|+1)}{n-1}.
\end{align*}
Since $\Add$ is a solution, $z$ must have greater temporal closeness centrality than $\vs$, implying $\cclos(G \cup \Add, z)-\cclos(G \cup \Add, \vs) > 0$, which gives us:
\[
|\Add|+q_{\Add}+1 > |U|+k.
\]
Notice however that $|\Add| \leq k$ (since the evader's budget is $b = k$) and that $q_{\Add} \leq |U|$ (since $q_{\Add}$ is the number of nodes in $U$ at finite average latency from $z$).
Therefore, we must have $|\Add|=k$ and $q_{\Add}=|U|$, which means that after the addition of $\Add$ for every $u_j \in U$ we have a node $S_i \in S$ connected to both $z$ and $u_j$ (there is no other way for $z$ to have a finite average latency to $u_j$).
Consequently, every element of the universe $u_j \in U$ is covered by at least one set $S_i \in S^*$, as $z$ is connected with a node $S_i$ only if it belongs to $S^*$, and $S_i$ is connected with $u_j$ only if it contains $u_j$ in the Set Cover problem instance.
We showed that if there exists a solution to the constructed instance of the Temporal Hiding problem, then there also exists a solution to the given instance of the Set Cover problem.
This concludes the proof.
\end{proof}

\begin{theorem}
\label{thrm:npc-betweenness}
Temporal Hiding problem is NP-complete given the betweenness temporal centrality.
\end{theorem}

\begin{proof}
The problem is trivially in NP, since after adding and removing a given set of contacts we can computed the ranking of all nodes according to the betweenness temporal centrality in polynomial time.

We will now prove that the problem is NP-hard.
To this end, we will show a reduction from the NP-complete \textit{$3$-Set Cover} problem.
The decision version of this problem is defined by a universe, $U=\{u_1,\ldots,u_{|U|}\}$, a collection of sets $S=\{S_1,\ldots,S_{|S|}\}$ such that $\forall_i S_i \subset U \land |S_i|=3$, and an integer $k \in \N$ where the goal is to determine whether there exist $k$ elements of $S$ the union of which equals $U$.

Let $(U,S,k)$ be a given instance of the $3$-Set Cover problem.
Let us assume that $|U| \geq 6$, note that all instances where $|U| < 6$ can be solved in polynomial time.
We will now construct an instance of the Temporal Hiding problem.
First, let us construct a temporal network $G=(V,\CC,T)$ where:
\begin{itemize}
\item $V = \{ \vs, w, y, z \} \cup U \cup S \cup \bigcup_{i=1}^{|U|+k-1} \{x_i\}$,
\item $\CC = \{(z,w,0),(\vs,y,0)\} \cup \bigcup_{x_i \in V}\{(\vs,x_i,1)\} \cup \bigcup_{u_i \in S_j}\{(u_i,S_j,2)\}$,
\item $T=3$.
\end{itemize}
An example of the construction of the network $G$ is presented in Figure~\ref{fig:npc-betweenness}.

\begin{figure}[t]
\centering
\includegraphics[width=.6\linewidth]{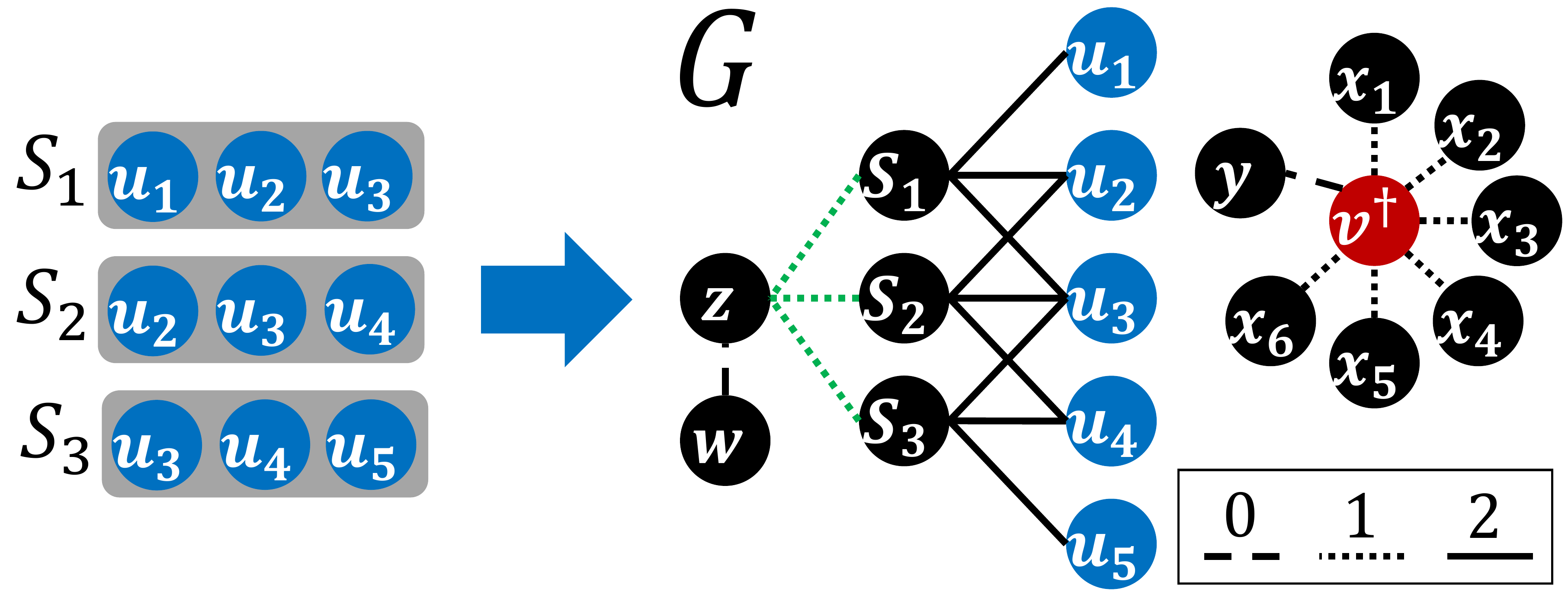}
\caption{
Construction used in the proof of Theorem~\ref{thrm:npc-betweenness}.
On the left there is an instance of the $3$-Set Cover problem, while on the right there is a network $G$ constructed as part of the Temporal Hiding problem instance.
The blue nodes in $G$ correspond to the elements of the universe $U$, while the evader is marked red.
Dashed lines correspond to contacts at moment $0$, dotted lines to contacts at moment $1$, and solid lines to contact at moment $2$.
Green lines depict the contacts allowed to be added.
}
\label{fig:npc-betweenness}
\end{figure}

Now, consider the instance $(G,\vs,\FA,\FR,b,c,\thr)$ of the Temporal Hiding problem, where:
\begin{itemize}
\item $G$ is the temporal network we just constructed,
\item $\vs$ is the evader,
\item $\FA = \{(z,S_i,1) : S_i \in V\}$,
\item $\FR = \emptyset$, i.e., none of the edges can be removed,
\item $b=k$,
\item $c$ is the temporal betweenness centrality,
\item $\thr=1$ is the safety threshold.
\end{itemize}

Since $\FR = \emptyset$, for any solution to the constructed instance of the Temporal Hiding problem we must have $\Rem = \emptyset$.
Hence, we will omit mentions of $\Rem$ in the remainder of the proof, and we will assume that a solution consists just of $\Add$.

First, let us analyze the temporal betweenness centrality values of the nodes in $G$ after the addition of any $A \subseteq \FA$.
Let $q_A$ denote the number of nodes $u_j \in U$ such that $\al(z,u_j)<\infty$, i.e., the number of nodes $u_j \in U$ such that $z$ is connected (via the contacts from $A$) with at least one node $S_i$ connected with $u_j$ (notice that there are no other possibilities for $z$ to have finite average latency to a node in $U$).
The temporal betweenness centrality values can be computed as follows:
\begin{itemize}
\item $\cbetw(G \cup A, \vs) = \frac{|U|+k-1}{(n-1)(n-2)T}$, as it belongs to the shortest temporal paths from $y$ to all $|U|+k-1$ nodes $x_i$,
\item $\cbetw(G \cup A, z) = \frac{|A|+q_A}{(n-1)(n-2)T}$, as it belongs to the shortest temporal paths from $w$ to all $|A|$ nodes $S_i$ it is connected with and to all $q_A$ nodes $u_i$ that can be reached from it,
\item $\cbetw(G \cup A, S_i) \leq \frac{6}{(n-1)(n-2)T} \leq \cbetw(G \cup A, \vs)$, for any $i$, as $S_i$ can belong to the shortest temporal paths from nodes in $\{z,w\}$ to the three nodes in $U$ that it is connected with,
\item $\cbetw(G \cup A, w) = \cbetw(G \cup A, y) = \cbetw(G \cup A, u_i) = \cbetw(G \cup A, x_i) = 0$, for any $i$, as none of these nodes belong to the shortest temporal paths between any pairs of other nodes.
\end{itemize}

Since the safety threshold is $\thr=1$, only one node has to have greater temporal betweenness centrality than $\vs$ after he addition of a given $A \subseteq \FA$ in order for said $A$ to be solution to the constructed instance of the Temporal Hiding problem.
However, notice that the only node that can have greater temporal betweenness centrality than $\vs$ (and therefore higher position in the ranking) is $z$.
In other words, a given $A \subseteq \FA$ is a solution to the constructed instance of the Temporal Hiding problem if and only if we have $\cbetw(G \cup A, z) > \cbetw(G \cup A, \vs)$.

We will now show that the constructed instance of the Temporal Hiding problem has a solution if and only if the given instance of the $3$-Set Cover problem has a solution.

Assume that there exists a solution to the given instance of the $3$-Set Cover problem, i.e., a subset $S^* \subseteq S$ of size $k$ the union of which is the universe $U$.
We will show that $\Add = \{z\} \times S^* \times \{1\}$ (i.e., connecting $z$ with nodes corresponding to all sets in $S^*$) is a solution to the constructed instance of the Temporal Hiding problem.
First, notice that after the addition of $\Add$ there exists a time-respecting path from $z$ to every node $u_j \in U$, leading through the node corresponding to an element $S_i \in S^*$ containing $u_j$, with which $z$ is now connected.
Hence, we have that $q_{\Add} = |U|$ and $|\Add|=k$, which gives us:
\begin{align*}
\cbetw(G \cup \Add, z) = \frac{|U|+k}{(n-1)(n-2)T} &> \frac{|U|+k-1}{(n-1)(n-2)T} \\
& = \cbetw(G \cup \Add, \vs).
\end{align*}
Therefore, after the addition of $\Add$ node $z$ has greater temporal betweenness centrality than the evader.
We showed that if there exists a solution to the given instance of the $3$-Set Cover problem, then there also exists a solution to the constructed instance of the Temporal Hiding problem.

Assume that there exists a solution $\Add$ to the constructed instance of the Temporal Hiding problem.
We will show that $S^* = \{S_i \in S : (z,S_i,1) \in \Add\}$ covers the universe $U$.
Let us compute the difference between $\cbetw(G \cup \Add, \vs)$ and $\cbetw(G \cup \Add, z)$:
\begin{align*}
\cbetw(G \cup \Add, z)-\cbetw&(G \cup \Add, \vs) \\
&= \frac{q_{\Add}+|\Add|-|U|-k+1}{(n-1)(n-2)T}.
\end{align*}
Since $\Add$ is a solution, $z$ must have greater temporal betweenness centrality than $\vs$, implying $\cbetw(G \cup \Add, z)-\cbetw(G \cup \Add, \vs) > 0$, which gives us:
\[
q_{\Add}+|\Add|+1 > |U|+k.
\]
Notice however that $|\Add| \leq k$ (since the evader's budget is $b = k$) and that $q_{\Add} \leq |U|$ (since $q_{\Add}$ is the number of nodes in $U$ at finite average latency from $z$).
Therefore, we must have $|\Add|=k$ and $q_{\Add}=|U|$, which means that after the addition of $\Add$ for every $u_j \in U$ we have a node $S_i \in S$ connected to both $z$ and $u_j$ (there is no other way for $z$ to have a finite average latency to $u_j$).
Consequently, every element of the universe $u_j \in U$ is covered by at least one set $S_i \in S^*$, as $z$ is connected with a node $S_i$ only if it belongs to $S^*$, and $S_i$ is connected with $u_j$ only if it contains $u_j$ in the $3$-Set Cover problem instance.
We showed that if there exists a solution to the constructed instance of the Temporal Hiding problem, then there also exists a solution to the given instance of the $3$-Set Cover problem.
This concludes the proof.
\end{proof}

\begin{algorithm}[tb!]
\caption{Finding a $2$-approximate solution for the Minimum Temporal Hiding problem given the degree temporal centrality.
}
\label{alg:appr-degree}
\begin{algorithmic}[1]

\small

\Input{Temporal networks $G=(V,\CC,T)$, evader $\vs$, set of edges allowed to be added $\FA$, set of edges allowed to be removed $\FR$, safety threshold $\thr$.}

\Output{Solution $(\Add,\Rem)$ to the instance $(G,\vs,\FA,\FR,\cclos,\thr)$ of the Temporal Hiding problem or $\perp$ if there is no solution.}

\State Let $\langle \!v \rangle_{i=1}^{n\!-\!1}\!$ be a sequence of all nodes in $\!V\!\!$ other than $\vs$
\State $\FA \gets \FA \setminus (\{\vs\}\!\times\!V\!\times\!\TT)$ \label{ln:alg1-2}
\State $\FR \gets \FR \cap (\{\vs\}\!\times\!V\!\times\!\TT)$ \label{ln:alg1-3}
\For {$z \gets 0, \ldots, |\FR|$} \label{ln:alg1-4}
	\State $X_i^{z}[\theta,q] \gets \infty$ for every $i,\theta,q$
	\State $X_0^{z}[0,0] \gets 0$
	\For {$i \gets 1, \ldots, n-1$} \label{ln:alg1-7}
		\State $\phi_i \gets \left|\FA \cap \left( \{v_i\}\!\times\!V\!\times\!\TT\!\right)\right|$
		\For {$r \gets 0, \ldots, \left| \FR \cap \left( \{\vs\}\!\times\!\{v_i\}\!\times\!\TT \right) \right|$} \label{ln:alg1-8}
			\State $a \gets \max \left( 0, |\Nb_G(\vs)|\!-\!z\!+\!1\!-\!|\Nb_G(v_i)|\!+\!r \right)\!$ \label{ln:alg1-9}
			\For {$\theta \gets 0, \ldots, \min(i-1,\thr)$} \label{ln:alg1-10}
				\For {$q \gets 0, \ldots, z$} \label{ln:alg1-11}
					\If {$a\!\!\leq\!\!\phi_i\!\land\!X_{i\!-\!1}^{z}\![\theta,\!q]\!\!+\!\!a\!\!<\!\!X_i^{z}\![\theta\!\!+\!\!1,\!q\!\!+\!\!r]$}$\!\!\!$ \label{ln:alg1-12}
						\State $X_i^{z}[\theta+1,q+r] \gets X_{i\!-\!1}^{z}[\theta,q]+a$ \label{ln:alg1-13}
						\State $Y_i^{z}[\theta+1,q+r] \gets (a,r)$ \label{ln:alg1-14}
					\EndIf
					\If {$a\!>\!0 \land X_{i\!-\!1}^{z}[\theta,\!q]\!<\!X_i^{z}[\theta,\!q\!+\!r]$} \label{ln:alg1-15}
						\State $X_i^{z}[\theta,q+r] \gets X_{i\!-\!1}^{z}[\theta,q]$ \label{ln:alg1-16}
						\State $Y_i^{z}[\theta,q+r] \gets (0,r)$ \label{ln:alg1-17}
					\EndIf
				\EndFor				
			\EndFor
		\EndFor
	\EndFor
	\State $C_z \gets X_{n-1}^z[\thr,z] + 2z$ \label{ln:alg1-18}
\EndFor
\State $z^* \gets \argmin_{z}C_z$ \label{ln:alg1-19}
\If {$C_{z^*} = \infty$} \label{ln:alg1-20}
	\State \Return $\perp$ \label{ln:alg1-21}
\EndIf
\State $(\Add,\Rem) \gets (\emptyset, \emptyset)$ \label{ln:alg1-22}
\State $x^* \gets C_{z^*}\!-\!2z^*$
\State $\theta^* \gets \thr$
\State $q^* \gets z^*$
\For {$i \gets n-1, \ldots, 1$}
	\State $(a,r) \gets Y_i^{z^*}[\theta^*,q^*]$
	\State $\Add \gets \Add \cup \select \left( a, \FA \cap \left( \{v_i\}\!\times\!V\!\times\!\TT \right) \right)$
	\State $\Rem \gets \Rem \cup \select \left( r, \FR \cap \left( \{\vs\}\!\times\!\{v_i\}\!\times\!\TT\right) \right)$
	\State $x^* \gets x^*\!-\!a$
	\State $q^* \gets q^*\!-\!r$
	\If {$|\Nb_G(v_i)|\!+\!a\!-\!r > |\Nb_G(\vs)|\!-\!z^*$}
		\State $\theta^* \gets \theta^*\!-\!1$
	\EndIf
\EndFor
\State \Return $(\Add,\Rem)$ \label{ln:alg1-34}
\end{algorithmic}
\end{algorithm}

\begin{theorem}
\label{thrm:appr-degree}
Algorithm~\ref{alg:appr-degree} is a $2$-approximation for the Minimum Temporal Hiding problem given the degree temporal centrality.
The bound is tight.
\end{theorem}

\begin{proof}
First, notice that removing contacts that are not incident with $\vs$ and adding edges that are incident with $\vs$ is always suboptimal, as in case of both these changes the difference between the degree of the evader and all other nodes either increases or stays the same.
Hence, no solution of the optimal size will include these types of modification and we exclude them in lines \ref{ln:alg1-2}-\ref{ln:alg1-3}.

Let the cost of a given solution be expressed by a number of stubs (half-contacts), i.e., twice the number of contacts.
In lines \ref{ln:alg1-4}-\ref{ln:alg1-19} the algorithm computes the minimal cost of a solution that satisfies the safety threshold, under assumption that we are allowed to connect any two stubs from contacts appearing in $\FA$ (thus we are guaranteed that the actual optimum is greater or equal).
It does so using the dynamic programming technique.
The loop in line~\ref{ln:alg1-4} iterates over the number of contacts removed from the network.
The algorithm fills the arrays $X$ and $Y$ in order to identify $C_z$, the optimal cost of solution that satisfies the safety thresholds while removing $z$ contacts from the network.
The element $X^z_i[\theta,q]$ is minimal number of stubs that have to added to nodes $v_1,\ldots,v_i$ so that $\theta$ of them have greater degrees than the evader, assuming that we removed $q$ stubs incident with them and that the degree of the evader is $|\Nb_G(\vs)|-z$.
The element $Y^z_i[\theta,q]$ stores the number of stubs that are added to and removed from node $v_i$ to achieve the solution with the cost $X^z_i[\theta,q]$.
The loop in line~\ref{ln:alg1-7} iterates over all nodes in the network other than the evader, while the loop in line~\ref{ln:alg1-8} iterates over the number of stubs removed from the node $v_i$.
The value $a$ computed in line is the number of stubs that have to added to node $v_i$ in order for it to contribute towards the safety threshold, i.e., in order for it to have greater degree than the evader.
The loops in lines~\ref{ln:alg1-10} and~\ref{ln:alg1-11} iterate over all values of $\theta$ and $q$ valid for arrays $X^z_{i-1}$ and $Y^z_{i-1}$, and based on them fill arrays $X^z_i$ and $Y^z_i$.
If the addition of the required number of $a$ stubs is possible (test in line~\ref{ln:alg1-12}) then we record a solution where the node $v_i$ contributes towards the safety threshold (value $\theta+1$ in lines~\ref{ln:alg1-13}-\ref{ln:alg1-14}).
In lines~\ref{ln:alg1-16}-\ref{ln:alg1-17}) we record a solution where the node $v_i$ does not contribute towards the safety threshold, but is only used to decrease the degree of the evader.
We do it only if counting towards the threshold actually required adding any stubs (test in line~\ref{ln:alg1-15}).
Finally, in line~\ref{ln:alg1-18} we compute the cost of the optimal solution that removes $z$ contacts from the network, while satisfying the safety threshold $\thr$.

We showed that $C_{z^*}$ is the minimal cost of a solution that satisfies the safety threshold, under assumption that we are allowed to connect any two stubs from contacts appearing in $\FA$.
If no solution was found (which is tested in line~\ref{ln:alg1-20}), then the algorithm returns $\perp$ in line~\ref{ln:alg1-21}.
Otherwise, the algorithm constructs the solution in lines~\ref{ln:alg1-22}-\ref{ln:alg1-34}.
Since we removed all contact including $\vs$ from $\FA$, any solution that removes all contacts from the optimal solution and adds all the stubs from the optimal solution satisfies the safety threshold.
What is more, such a solution contains at most twice as many contacts as the optimum (as it is possible that each stub in $\FA$ from the optimal solution gets connected to a stub not appearing in the optimal solution).
Hence, the algorithm is a $2$-approximation.
The complexity of the algorithm is $\mathcal{O}(|\FR|^3 n \thr)$.

\begin{figure}[t]
\centering
\includegraphics[width=.4\linewidth]{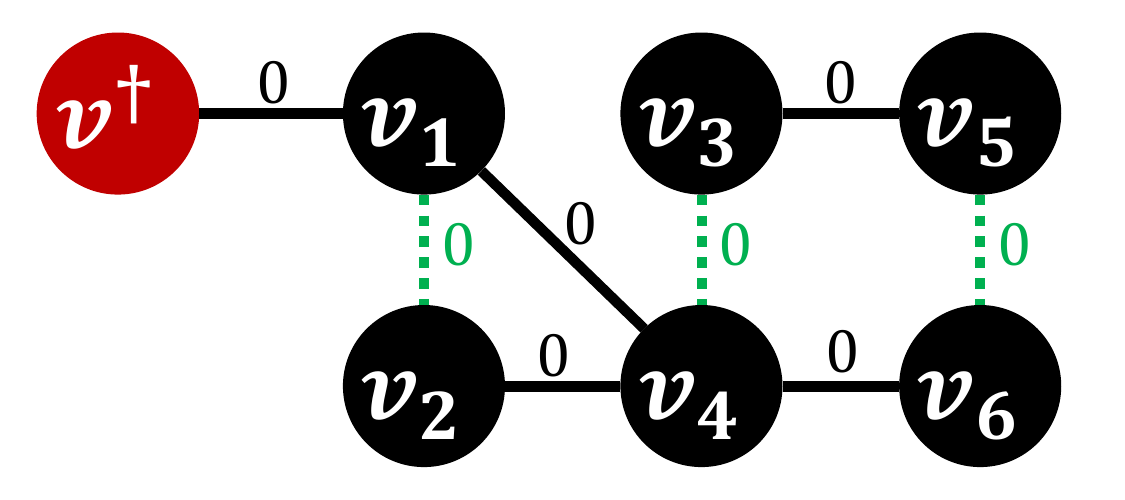}
\caption{
Example of a network for which the approximation ratio of Algorithm~\ref{alg:appr-degree} is exactly $2$.
The red node corresponds to the evader.
Numbers next to edges are the moments in which the contacts occur.
Green dotted lines depict the contacts allowed to be added.
}
\label{fig:appr-degree}
\end{figure}

Finally, to prove the claim about the tight bound, Figure~\ref{fig:appr-degree} presents an example of a network for which the approximation ratio of Algorithm~\ref{alg:appr-degree} is exactly $2$.
Given a safety threshold $\thr=4$, Algorithm~\ref{alg:appr-degree} will add to the network contacts $(v_1,v_2,0)$ and $(v_3,v_4,0)$ (because of the order of nodes in sequence $\langle v \rangle_{i=1}^{n-1}$), whereas the optimal solution is to add the contact $(v_5,v_6,0)$.
This concludes the proof.
\end{proof}

\begin{theorem}
\label{thrm:appr-closeness}
The Minimum Temporal Hiding problem given the closeness temporal centrality cannot be approximated within a ratio of $(1-\epsilon) \ln n$ for any $\epsilon > 0$, unless $P=NP$.
\end{theorem}

\begin{proof}
In order to prove the theorem, we will use the result by Dinur and Steurer~\cite{dinur2014analytical} that the Minimum Set Cover problem cannot be approximated within a ratio of $(1-\epsilon) \ln n$ for any $\epsilon > 0$, unless $P=NP$.

Let $X=(U,S)$ be an instance of the Minimum Set Cover problem, where $U$ is the universe $\{u_1, \ldots, u_{|U|}\}$, and $S$ is a collection $\{S_1, \ldots, S_{|S|}\}$ of subsets of $U$.
The goal here is to find subset $S^* \subseteq S$ such that the union of $S^*$ equals $U$ and the size of $S^*$ is minimal.

First, we will show a function $f(X)$ that based on an instance of the Minimum Set Cover problem $X=(U,S)$ constructs an instance of the Minimum Temporal Hiding problem.
Let us assume that $|S| \geq 3$, note that all instances where $|S| < 3$ can be solved in polynomial time.

Let the temporal network $G(X)=(V,\CC,T)$ be defined as:
\begin{itemize}
\item $V = \{ \vs, z \} \cup U \cup S \cup \bigcup_{i=1}^{|U|+|S|-1} \{x_i\}$,
\item $\CC = \bigcup_{x_i \in V}\{(\vs,x_i,0)\} \cup \bigcup_{S_i \in V}\{(z,S_i,1)\} \cup \bigcup_{u_i \in S_j}\{(u_i,S_j,1)\}$,
\item $T=2$.
\end{itemize}
An example of the construction of the network $G(X)$ is presented in Figure~\ref{fig:appr-closeness}.

\begin{figure}[t]
\centering
\includegraphics[width=.6\linewidth]{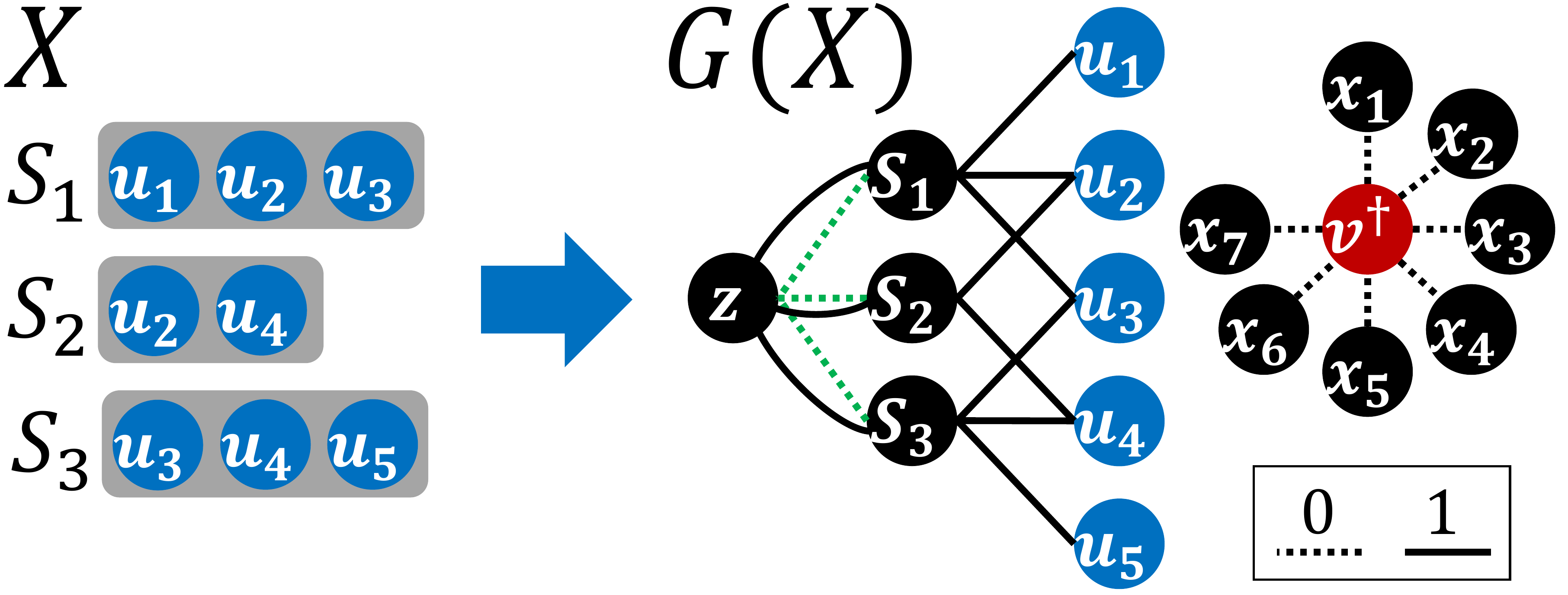}
\caption{
Construction used in the proof of Theorem~\ref{thrm:appr-closeness}.
On the left there is an instance $X$ of the Minimum Set Cover problem, while on the right there is a network $G(X)$ constructed as part of the Minimum Temporal Hiding problem instance.
The blue nodes in $G(X)$ correspond to the elements of the universe $U$, while the evader is marked red.
Dotted lines correspond to contacts at moment $0$, and solid lines to contact at moment $1$.
Green lines depict the contacts allowed to be added.
}
\label{fig:appr-closeness}
\end{figure}

The formula of function $f$ is then $f(X)=(G(X),\vs,\FA,\FR,c,\thr)$, where:
\begin{itemize}
\item $G(X)$ is the temporal network we just constructed,
\item $\vs$ is the evader,
\item $\FA = \{(z,S_i,0) : S_i \in V\}$,
\item $\FR = \emptyset$, i.e., none of the edges can be removed,
\item $c$ is the temporal closeness centrality,
\item $\thr=1$ is the safety threshold.
\end{itemize}

Let $\Add$ be the solution to the constructed instance of the Minimum Temporal Hiding problem $f(X)$ (notice that since $\FR=\emptyset$, we must have $\Rem=\emptyset$).
The function $g$ computing corresponding solution to the instance $X$ of the Minimum Set Cover problem is now $g(X,\Add)= \{S_i \in S : (z,S_i,0) \in \Add\}$, i.e., $S^*$ is the set of all sets $S_i$ such that their corresponding nodes $S_i$ are connected with $z$ via the contacts in $\Add$.

Now, we will show that $g(X,\Add)$ is indeed a correct solution to $X$, i.e., that it covers the entire universe.
Let $q_A$ denote the number of nodes $u_j \in U$ such that $\al_{G(X) \cup A}(z,u_j)<\infty$ for a given set of contacts $A \subseteq \FA$.
Centrality values after adding $A$ to the network $G(X)$ are as follows:
\begin{itemize}
\item $\cclos(G(X) \cup A, \vs) = \frac{2(|U|+|S|-1)}{n-1}$,
\item $\cclos(G(X) \cup A, z) = \frac{2(|S|+q_A)}{n-1}$,
\item $\cclos(G(X) \cup A, S_i) \leq \frac{2(|U|+1)}{n-1} < \cclos(G(X) \cup A, \vs)$ for every $S_i \in V$,
\item $\cclos(G(X) \cup A, u_i) = \frac{2(|\{S_j \in S: u_i \in S_j\}|)}{n-1} \leq \frac{2|S|}{n-1} < \cclos(G(X) \cup A, \vs)$ for every $u_i \in V$,
\item $\cclos(G(X) \cup A, x_i) = \frac{2}{n-1} < \cclos(G(X) \cup A, \vs)$ for every $x_i \in V$.
\end{itemize}
Therefore, the safety threshold $\thr$ is satisfied if and only if $\cclos(G(X) \cup A, z) > \cclos(G(X) \cup A, \vs)$, which is the case when $q_A = |U|$.
Hence, if $\Add$ is a solution to $f(X)$, then for every node $u_j$ there exists a node $S_i$ such that $(z,S_i,0)\in \Add$ and $u_j$ is connected with $S_i$.
However, because of the way we constructed the network $G(X)$, this can only be the case when $u_j \in S_i$ in the given instance $X$ of the Minimum Set Cover problem.
Therefore, for every element of the universe $u_j$ there exists a set $S_i$ such that $u_j \in S_i$ and $S_i$ in $g(X,\Add)$, which implies that $g(X,\Add)$ is a solution to the given instance $X$ of the Minimum Set Cover problem, i.e., it covers the universe.
What is more, since $|g(X,\Add)|=|\Add|$, we also have that the optimal solutions to both instances are of the same size.

Now, assume that there exists an $r$-approximation algorithm for the Minimum Temporal Hiding problem where $r=(1-\epsilon)\ln |\FA|$ for some $\epsilon > 0$.
Let us use this algorithm to solve the constructed instance $f(X)$ and consider solution $g(X,\Add)$ to the given instance $X$ of the Minimum Set Cover problem.
Since the size of the optimal solution is the same for both instances, we obtained an approximation algorithm that solves Minimum Set Cover problem to within $(1-\epsilon)\ln n$ for $\epsilon > 0$.
However, Dinur and Steurer~\cite{dinur2014analytical} showed that the Minimum Set Cover problem cannot be approximated within a ratio of $(1-\epsilon) \ln n$ for any $\epsilon > 0$, unless $P=NP$.
Therefore, such approximation algorithm for the Minimum Temporal Hiding problem cannot exist, unless $P=NP$.
This concludes the proof.
\end{proof}

\begin{theorem}
\label{thrm:appr-betweenness}
The Minimum Temporal Hiding problem given the betweenness temporal centrality cannot be approximated within a ratio of $(1-\epsilon) \ln n$ for any $\epsilon > 0$, unless $P=NP$.
\end{theorem}

\begin{proof}
In order to prove the theorem, we will use the result by Dinur and Steurer~\cite{dinur2014analytical} that the Minimum Set Cover problem cannot be approximated within a ratio of $(1-\epsilon) \ln n$ for any $\epsilon > 0$, unless $P=NP$.

Let $X=(U,S)$ be an instance of the Minimum Set Cover problem, where $U$ is the universe $\{u_1, \ldots, u_{|U|}\}$, and $S$ is a collection $\{S_1, \ldots, S_{|S|}\}$ of subsets of $U$.
The goal here is to find subset $S^* \subseteq S$ such that the union of $S^*$ equals $U$ and the size of $S^*$ is minimal.

First, we will show a function $f(X)$ that based on an instance of the Minimum Set Cover problem $X=(U,S)$ constructs an instance of the Minimum Temporal Hiding problem.
Let us assume that $\forall_i |S_i| < |U|$, note that all instances where there exists $S_i = U$ can be solved in polynomial time.

Let the temporal network $G(X)=(V,\CC,T)$ be defined as:
\begin{itemize}
\item $V = \{ \vs, z \} \cup U \cup S \cup \bigcup_{i=1}^{|U|} \{w_i\} \cup \bigcup_{i=1}^{|U|+1} \{x_i\} \cup \bigcup_{i=1}^{|U|-1} \{y_i\}$,
\item $\CC = \bigcup_{x_i \in V}\{(\vs,x_i,0)\} \cup \bigcup_{y_i \in V}\{(\vs,y_i,2)\}
\cup \bigcup_{w_i \in V}\{(z,w_i,0)\} \cup \bigcup_{w_i,w_j \in V}\{(w_i,w_j,2)\}$\\ $\cup \bigcup_{w_i,S_j \in V}\{(w_i,S_j,2)\} \cup \bigcup_{S_i \in V}\{(z,S_i,2)\} \cup \bigcup_{S_i,S_j \in V}\{(S_i,S_j,2)\} \cup \bigcup_{u_i \in S_j}\{(u_i,S_j,2)\}$,
\item $T=3$.
\end{itemize}
An example of the construction of the network $G(X)$ is presented in Figure~\ref{fig:appr-betweenness}.

\begin{figure}[t]
\centering
\includegraphics[width=.6\linewidth]{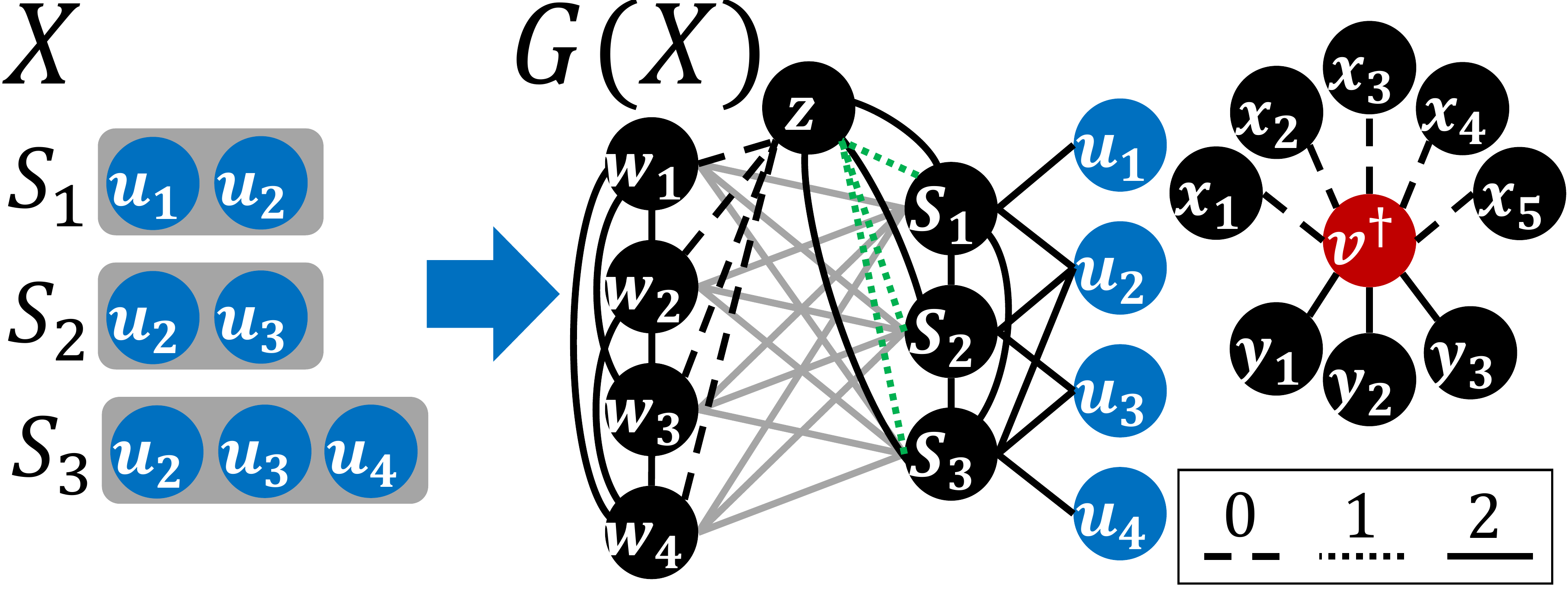}
\caption{
Construction used in the proof of Theorem~\ref{thrm:appr-betweenness}.
On the left there is an instance $X$ of the Minimum Set Cover problem, while on the right there is a network $G(X)$ constructed as part of the Minimum Temporal Hiding problem instance.
The blue nodes in $G(X)$ correspond to the elements of the universe $U$, while the evader is marked red.
Dashed lines correspond to contacts at moment $0$, dotted lines to contacts at moment $1$, and solid lines to contact at moment $2$.
Green lines depict the contacts allowed to be added.
Contacts between nodes $S_i$ and $w_j$ are depicted grey for better readability.
}
\label{fig:appr-betweenness}
\end{figure}

The formula of function $f$ is then $f(X)=(G(X),\vs,\FA,\FR,c,\thr)$, where:
\begin{itemize}
\item $G(X)$ is the temporal network we just constructed,
\item $\vs$ is the evader,
\item $\FA = \{(z,S_i,1) : S_i \in V\}$,
\item $\FR = \emptyset$, i.e., none of the edges can be removed,
\item $c$ is the temporal betweenness centrality,
\item $\thr=1$ is the safety threshold.
\end{itemize}

Let $\Add$ be the solution to the constructed instance of the Minimum Temporal Hiding problem $f(X)$ (notice that since $\FR=\emptyset$, we must have $\Rem=\emptyset$).
The function $g$ computing corresponding solution to the instance $X$ of the Minimum Set Cover problem is now $g(X,\Add)= \{S_i \in S : (z,S_i,0) \in \Add\}$, i.e., $S^*$ is the set of all sets $S_i$ such that their corresponding nodes $S_i$ are connected with $z$ via the contacts in $\Add$.

Now, we will show that $g(X,\Add)$ is indeed a correct solution to $X$, i.e., that it covers the entire universe.
Let $q_A$ denote the number of nodes $u_j \in U$ such that $\al_{G(X) \cup A}(z,u_j)<\infty$ for a given set of contacts $A \subseteq \FA$.
Centrality values after adding $A$ to the network $G(X)$ are as follows:
\begin{itemize}
\item $\cbetw(G(X) \cup A, \vs) = \frac{(|U|+1)(|U|-1)}{(n-1)(n-2)T} = \frac{|U|^2-1}{(n-1)(n-2)T}$, as $\vs$ controls all shortest temporal paths from nodes $x_i$ to $y_i$,
\item $\cbetw(G(X) \cup A, z) = \frac{q_A|U|}{(n-1)(n-2)T}$, as $z$ controls all shortest temporal paths from nodes $w_i$ to nodes $u_j$ that are reachable from $z$,
\item $\cbetw(G(X) \cup A, S_i) \leq \frac{(|U|+1)(|U|-1)}{(n-1)(n-2)T} = \frac{|U|^2-1}{(n-1)(n-2)T} = \cbetw(G(X) \cup A, \vs)$, as $S_i$ can control only paths from nodes $w_i$ and $z$ (there are $|U|+1$ such nodes) to nodes $u_j$ that $S_i$ has contact with (we made an assumption that there are at most $|U|-1$ such nodes),
\item $\cbetw(G(X) \cup A, w_i) = \cbetw(G(X) \cup A, u_i) = \cbetw(G(X) \cup A, x_i) = \cbetw(G(X) \cup A, y_i) = 0 < \cbetw(G(X) \cup A, \vs)$, as these node do not control any shortest paths.
\end{itemize}

Therefore, the safety threshold $\thr$ is satisfied if and only if $\cbetw(G(X) \cup A, z) > \cbetw(G(X) \cup A, \vs)$, which is the case when $q_A = |U|$.
Hence, if $\Add$ is a solution to $f(X)$, then for every node $u_j$ there exists a node $S_i$ such that $(z,S_i,1)\in \Add$ and $u_j$ is connected with $S_i$.
However, because of the way we constructed the network $G(X)$, this can only be the case when $u_j \in S_i$ in the given instance $X$ of the Minimum Set Cover problem.
Therefore, for every element of the universe $u_j$ there exists a set $S_i$ such that $u_j \in S_i$ and $S_i$ in $g(X,\Add)$, which implies that $g(X,\Add)$ is a solution to the given instance $X$ of the Minimum Set Cover problem, i.e., it covers the universe.
What is more, since $|g(X,\Add)|=|\Add|$, we also have that the optimal solutions to both instances are of the same size.

Now, assume that there exists an $r$-approximation algorithm for the Minimum Temporal Hiding problem where $r=(1-\epsilon)\ln |\FA|$ for some $\epsilon > 0$.
Let us use this algorithm to solve the constructed instance $f(X)$ and consider solution $g(X,\Add)$ to the given instance $X$ of the Minimum Set Cover problem.
Since the size of the optimal solution is the same for both instances, we obtained an approximation algorithm that solves Minimum Set Cover problem to within $(1-\epsilon)\ln n$ for $\epsilon > 0$.
However, Dinur and Steurer~\cite{dinur2014analytical} showed that the Minimum Set Cover problem cannot be approximated within a ratio of $(1-\epsilon) \ln n$ for any $\epsilon > 0$, unless $P=NP$.
Therefore, such approximation algorithm for the Minimum Temporal Hiding problem cannot exist, unless $P=NP$.
This concludes the proof.
\end{proof}

\section{Empirical Analysis}

In this section, we present the experimental analysis of the problem using simulations.

\begin{table}[t!]
\caption{Datasets considered in the simulations.
}
\label{tab:datasets}
\centering
\small
\begin{tabular}{lcc}
\hline
Network & $|V|$ & $|\CC|$ \\
\hline
Bluetooth~\cite{madan2012sensing} & 74 & 87491 \\
Call center~\cite{reality_badges} & 52 & 1182 \\
Cambridge~\cite{cambridge} & 187 & 8769 \\
Conference 1~\cite{isella2011s} & 113 & 8218 \\
Conference 2~\cite{cambridge} & 199 & 27165 \\
Conference 3~\cite{stehle2011simulation} & 402 & 28954 \\
Conference 4~\cite{stehle2011simulation} & 361 & 19304 \\
%Copenhagen Bluetooth~\cite{sapiezynski2019interaction} & 672 & 128266 \\
Copenhagen call~\cite{stopczynski2014measuring} & 484 & 1667 \\
Copenhagen SMS~\cite{stopczynski2014measuring} & 535 & 6165 \\
Diary~\cite{read2008dynamic} & 49 & 1427 \\
High school 1~\cite{mastrandrea2015contact} & 312 & 28780 \\
High school 2~\cite{mastrandrea2015contact} & 310 & 35592 \\
High school 3~\cite{mastrandrea2015contact} & 303 & 30383 \\
High school 4~\cite{mastrandrea2015contact} & 295 & 28112 \\
High school 5~\cite{mastrandrea2015contact} & 299 & 26391 \\
Hospital~\cite{vanhems2013estimating} & 75 & 9313 \\
Hospital colocation~\cite{vanhems2013estimating} & 73 & 35200 \\
Intel~\cite{cambridge} & 113 & 7408 \\
Kenya~\cite{kiti2016quantifying} & 52 & 2070 \\
Office~\cite{NWS:9950811} & 92 & 2679 \\
Office colocation~\cite{NWS:9950811} & 95 & 45357 \\
Polish manufacturer email~\cite{michalski2011matching} & 167 & 43360 \\
Primary school 1~\cite{stehle2011high} & 236 & 52339 \\
Primary school 2~\cite{stehle2011high} & 238 & 56313 \\
Reality~\cite{eagle2006reality} & 64 & 4251 \\
Romania~\cite{marin2012exploring} & 42 & 19045 \\
St Andrews~\cite{scott2006crawdad} & 25 & 1483 \\
Undergrads call~\cite{madan2012sensing} & 75 & 3574 \\
Undergrads SMS~\cite{madan2012sensing} & 41 & 758 \\
University couples call~\cite{aharony2011socialfmri} & 126 & 31155 \\
University couples SMS~\cite{aharony2011socialfmri} & 110 & 11962 \\
\hline
\end{tabular}
\end{table}

\subsection{Datasets}

In our simulations we consider a number of real-life temporal network datasets.
All datasets are presented in Table~\ref{tab:datasets}.

Given that the time resolutions of different datasets are vastly different, we perform the following normalization procedure.
We divide the time interval between the first and the last timestamp in the original dataset into $1000$ equal subintervals.
For each pair of nodes we consider them to have contact at time step $t$ if and only if the $t$-th subinterval contains at least one contact in the original dataset.
The length of the time interval of such normalized dataset is then $T = 1000$.

The third column of Table~\ref{tab:datasets} contains the number of contacts after performing the normalization procedure.
Such normalized datasets are then used in the simulations.

\begin{figure*}[t!]
\centering
\includegraphics[width=.9\textwidth]{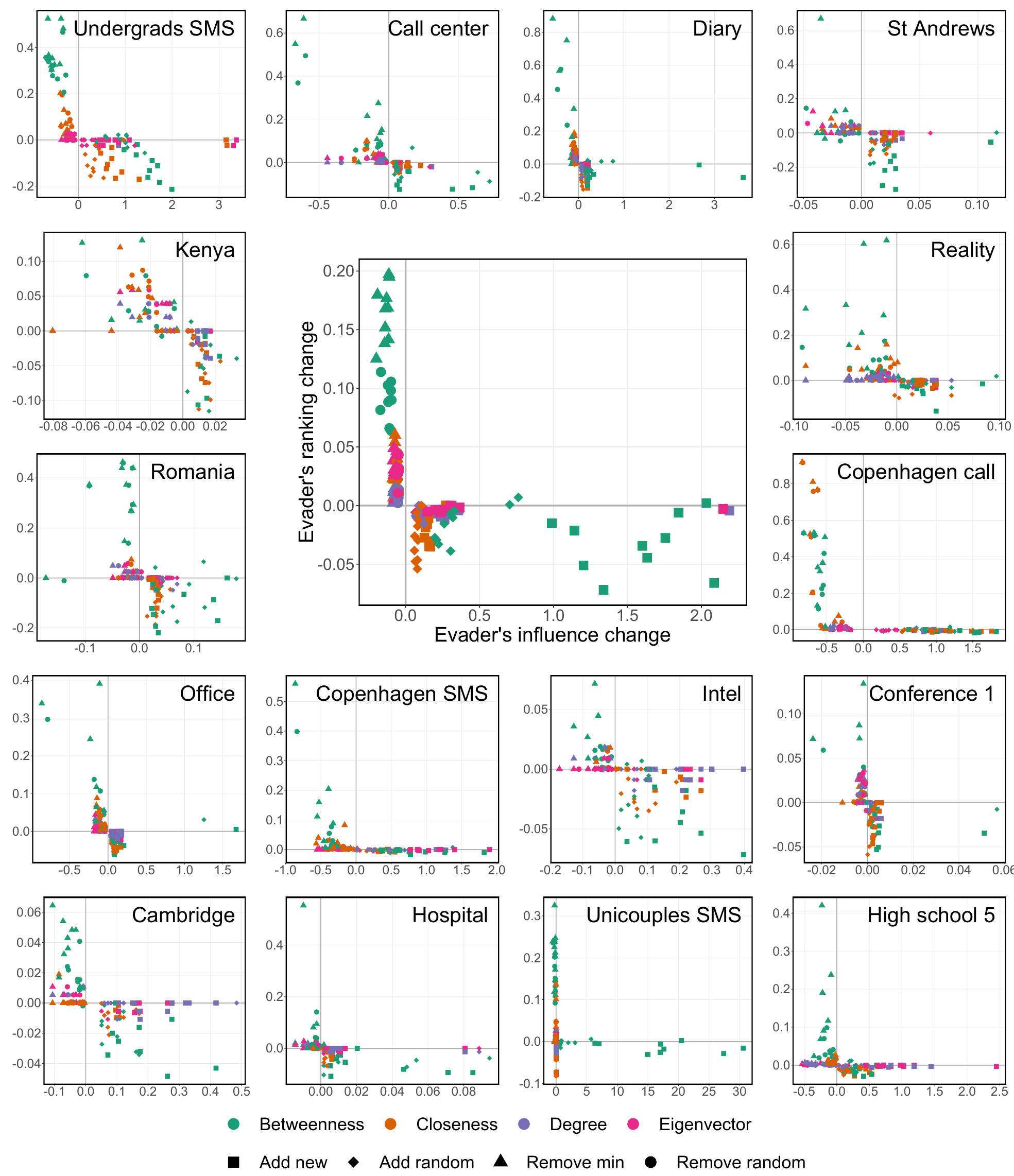}
\caption{
Effects of hiding the top $10$ nodes in each centrality ranking.
Smaller plots present results for different real-life networks. The large central plot presents data averaged over all networks (with each point corresponding to an average over nodes at a particular position in the ranking of a given centrality). The y-axis corresponds to the change in the evader's centrality ranking relative to the number of all nodes (with positive values indicating that the evader becomes better hidden due to the heuristic). The x-axis corresponds to the evader's influence relative to their initial influence. Different colors correspond to centrality measures, while various shapes correspond to heuristics.
}
\label{fig:scatter}
\end{figure*}

\begin{table*}[t!]
\caption{Temporal node descriptors.}
\label{tab:node-descriptors}
\centering
\small
\begin{tabular}{ccl}
\hline
\textbf{Type} &
\textbf{Descriptor} &
\textbf{Definition} \\
\hline
% Time evolution
& $\phi_C$ & Fraction of node's contacts when half of all contacts happened \\
& $\phi_T$ & Fraction of node's contacts at half the time $T$  \\
& $\rho_C$ & Fraction of node's edges present when half of all contacts happened \\
& $\rho_T$ & Fraction of node's edges present at half the time $T$ \\
& $\rho_c$ & Fraction of node's edges present at both the first and last $5\%$ of all contacts \\
\multirow{-6}{*}{\makecell{Time\\evolution}} & $\rho_t$ & Fraction of node's edges present at both the first and last $5\%$ of the time $T$ \\
\hline
% Edge activity
& $\varepsilon_m$ & Mean of intercontact times over node's edges \\
& $\varepsilon_s$ & Standard deviation of intercontact times over node's edges \\
& $\varepsilon_v$ & Coefficient of variation of intercontact times over node's edges \\
& $\varepsilon_k$ & Skewness of intercontact times over node's edges \\
& $\lambda_m$ & Mean of the number of other contacts between two consecutive contacts \\
& & of a node's edge \\
& $\lambda_s$ & Standard deviation of the number of other contacts between two consecutive \\
& & contacts of a node's edge \\
& $\lambda_v$ & Coefficient of variation of the number of other contacts between two consecutive \\
& & contacts of a node's edge \\
& $\lambda_k$ & Skewness of the number of other contacts between two consecutive \\
\multirow{-12}{*}{\makecell{Edge\\activity}} & & contacts of a node's edge \\
\hline
% Node activity
&  $\nu_m$ & Like $\varepsilon_m$ but for node's contacts \\
&  $\nu_s$ & Like $\varepsilon_s$ but for node's contacts \\
&  $\nu_v$ & Like $\varepsilon_v$ but for node's contacts \\
&  $\nu_k$ & Like $\varepsilon_k$ but for node's contacts \\
& $\eta_m$ & Like $\lambda_m$ but for node's contacts \\
& $\eta_s$ & Like $\lambda^{c}_s$ but for node's contacts \\
& $\eta_v$ & Like $\lambda^{c}_v$ but for node's contacts \\
\multirow{-8}{*}{\makecell{Node\\activity}} & $\eta^{c}_k$ & Like $\lambda^{c}_k$ but for node's contacts \\
\hline
% Network structure
& $\delta$ & Degree of the node \\
& $\delta_R$ & Degree of the node divided by the average degree \\
& $\kappa$ & Number of contacts of the node \\
& $\kappa_R$ & Number of contacts of the node divided by the average number of contacts \\
\multirow{-5}{*}{\makecell{Network\\structure}} & $\zeta$ & Local clustering coefficient of the node \\
\hline
\end{tabular}
\end{table*}

\subsection{Heuristics}

As shown in Section~\ref{sec:theoretical-analysis}, the task of finding an optimal way to hide from temporal centralities is computationally intractable. Hence, we now consider a number of heuristic solutions that add or remove contacts from the network in hope of improving the concealment of the evader:

\begin{description}
\item[Remove Random] Uniformly at random selects a neighbor of the evader, then removes one of the contacts with this neighbor (chosen uniformly at random),
\item[Remove Minimum] Removes one of the contacts (chosen uniformly at random) with a neighbor with the smallest number of contacts with the evader,
\item[Add Random] Uniformly at random selects a neighbor of the evader, then adds another contact with this neighbor (the moment of which is chosen uniformly at random),
\item[Add New] Adds a contact (the moment of which is chosen uniformly at random) between the evader and a second degree neighbor of the evader (chosen uniformly at random), i.e., a node that has at least one common neighbor with the evader, but is not the evader's neighbor.
\end{description}

For both adding and removing contacts we consider two heuristics---one selecting contacts to add or remove uniformly at random, and one with a more strategic approach, focused either on removing all contacts with some neighbors, or on finding new neighbors.

\begin{figure*}[t!]
\centering
\setlength\tabcolsep{0pt}
\begin{tabular}{m{.4\textwidth}m{.4\textwidth}}
\includegraphics[width=\linewidth]{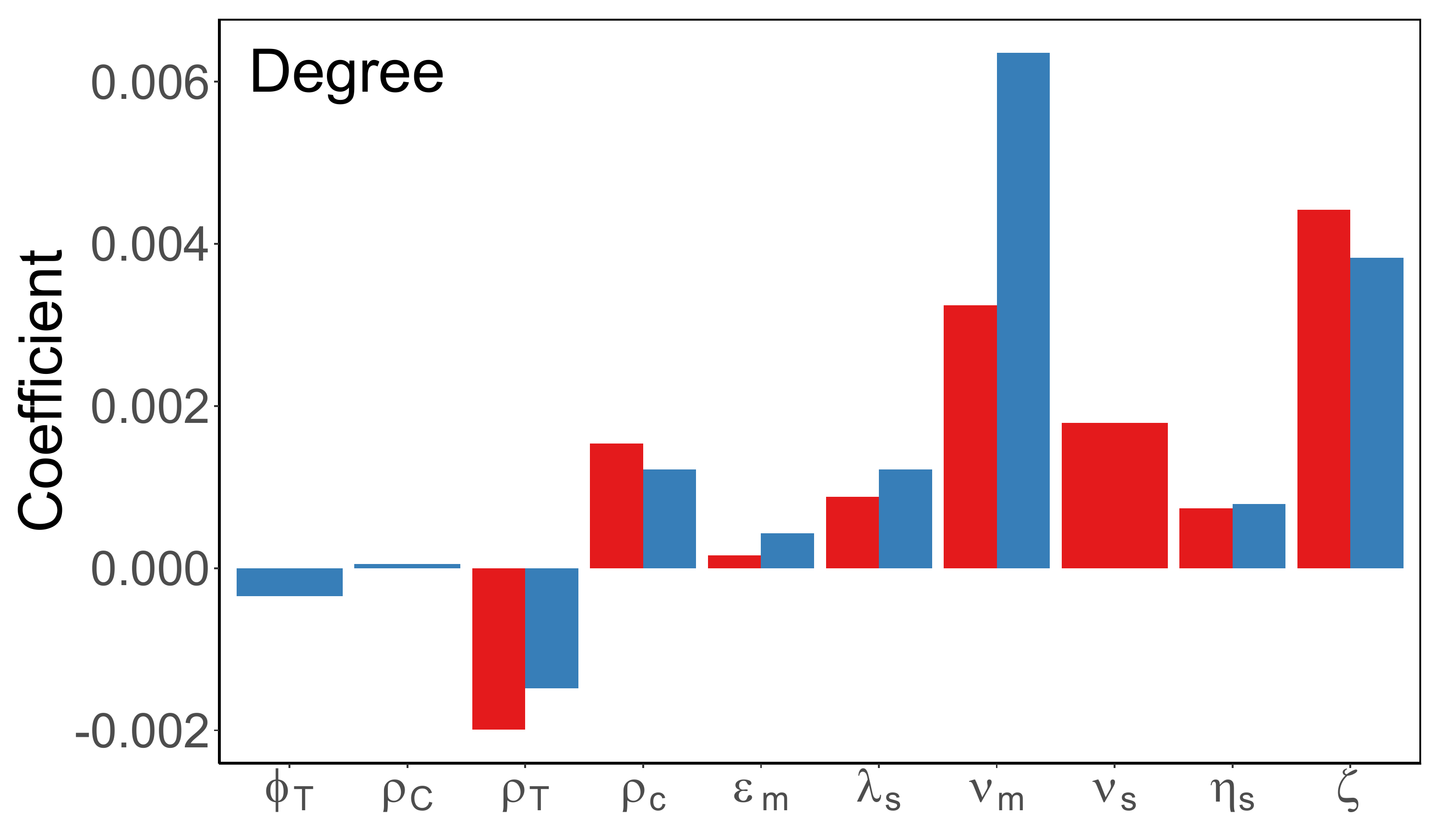} &
\includegraphics[width=\linewidth]{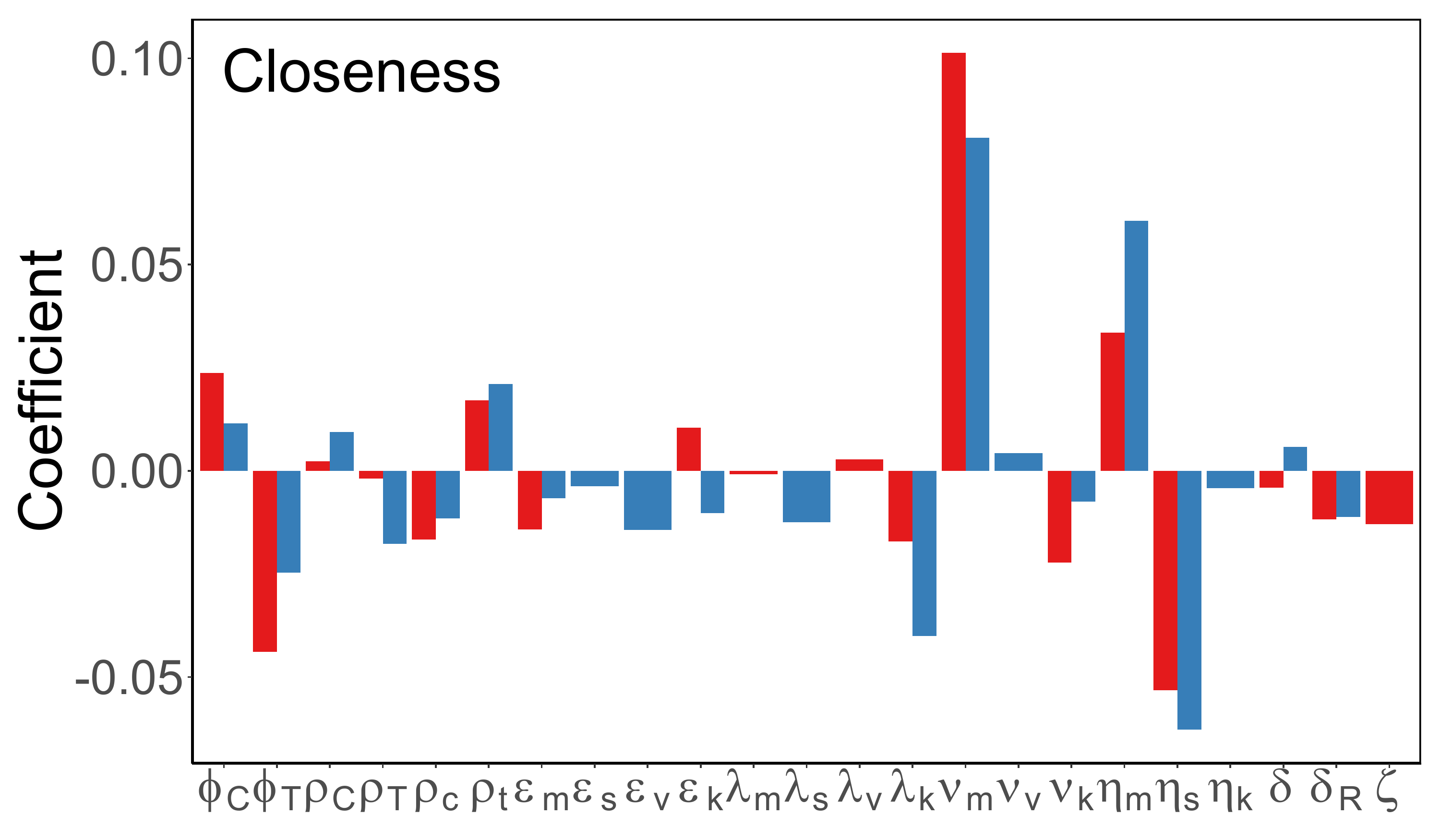} \\
\includegraphics[width=\linewidth]{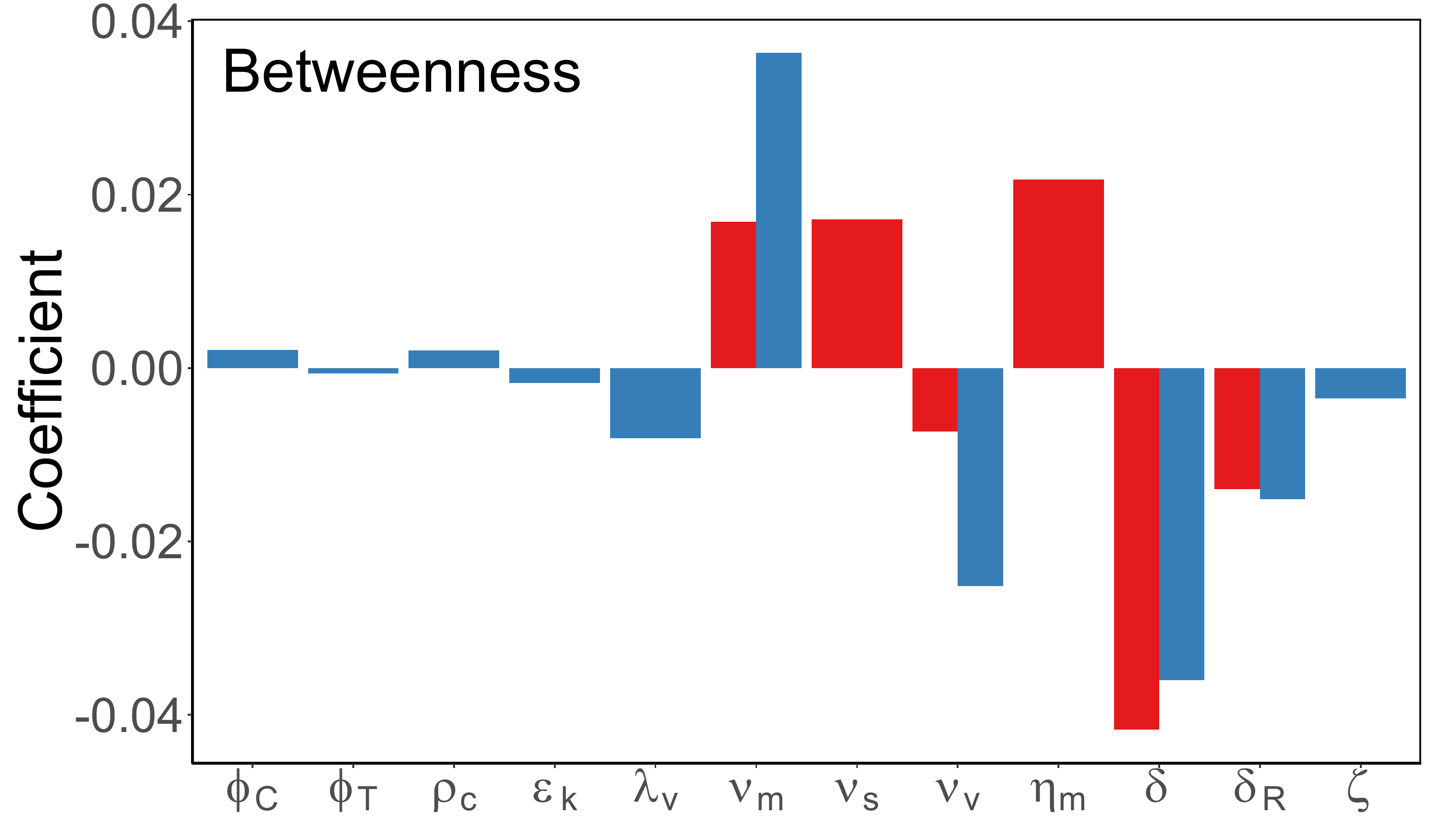} &
\includegraphics[width=\linewidth]{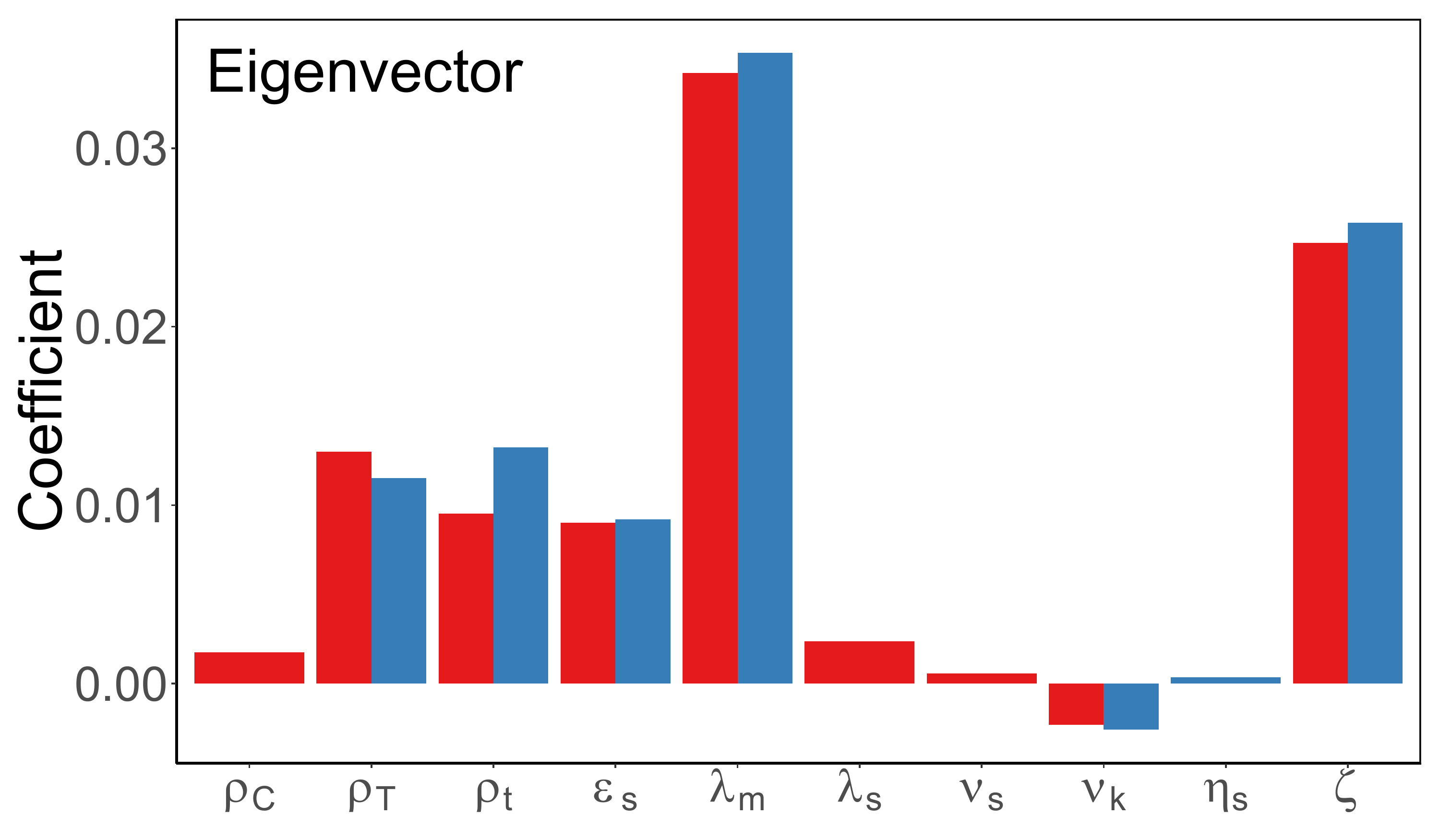}
\\
\multicolumn{2}{c}{\includegraphics[width=.4\linewidth]{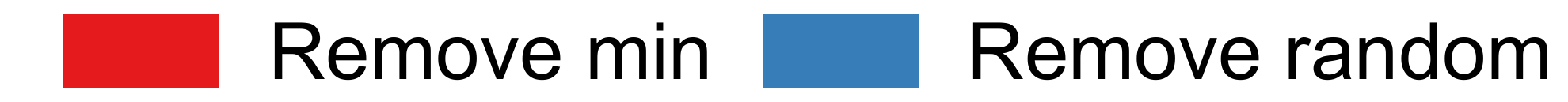}} \\
\end{tabular}
\caption{
Lasso regression coefficients of the node descriptors.
In each plot the x-axis corresponds to descriptors (defined in Table~\ref{tab:node-descriptors}) with non-zero coefficients in the Lasso regression, while the y-axis corresponds to the value of said coefficients. Different colors corresponds to different hiding heuristics.
}
\label{fig:node-lasso-wide}
\end{figure*}

\subsection{Basic Simulations}

Given a temporal network $G$ and a centrality measure $c$, we select $10$ top nodes in the centrality ranking as the potential evaders. We then attempt to hide each such evader $\vs$ using each of the considered heuristics described above. For each heuristic, we add or remove $5\%$ of the evader's contacts (at least ten for nodes with exceptionally few contacts).

We measure the following values before and after the hiding process:
\begin{itemize}
\item the position of $\vs$ in the ranking of $c$,
\item the centrality value of $\vs$ according to $c$,
\item the influence of $\vs$ over the network, measured as the average probability that $\vs$ gets infected if a Susceptible-Infected process starts in a different node times the expected number of infected nodes if a Susceptible-Infected process starts in $\vs$ (we consider the process with the probability of infection $10\%$).
\end{itemize}

We repeated the simulation $100$ times for each network and presented the results as averages.

Figure~\ref{fig:scatter} presents the results of our simulations.
As shown in the figure, in the vast majority of cases, the removal of contacts results in the evader's being more hidden, i.e., dropping in centrality ranking and becoming less influential. At the same time, adding contacts to the network has the opposite effects, i.e., the evader becomes more influential and more exposed to temporal centrality analysis. It suggests a centrality-influence trade-off in temporal networks. When comparing the effects of heuristics performing random changes, be it adding or removing contacts, with those performing strategic manipulation, i.e., Remove Minimum and Add New heuristics, we can see that using strategic heuristics result in a more pronounced effect, i.e., the magnitude of change in centrality ranking or influence value is greater than in the case of their random counterparts. Finally, when comparing different centrality measures, the betweenness centrality is, on average, the most sensitive to manipulation, i.e., executing the hiding process results in the greatest change in the betweenness centrality ranking. The next most sensitive centrality measure is the closeness centrality, with the degree and the eigenvector centrality measures being the most resilient on average.

\subsection{Regression Analysis}

The results of the simulations presented so far give us some insight into how effective different hiding heuristics can be. Still, they do not explain what qualities of the evader determine how successfully it can hide from temporal centralities. To this end, we now perform the regression analysis of the hiding process. Since the analysis presented in the previous section showed that only removal heuristics successfully hide the evader, we will focus our attention on them.

For each node considered an evader in our experiments, we compute several descriptors, i.e., its characteristics relating to its contacts and the network structure in which it is embedded. We consider four different categories of descriptors: time evolution, edge activity, node activity, and network structure. All descriptors are presented in Table~\ref{tab:node-descriptors}.

Figure~\ref{fig:node-lasso-wide} presents the regression coefficients computed using the Lasso regression methods~\cite{tibshirani1996regression}. Since the Lasso regression performs variable selection, not all descriptors appear in the figure (the coefficients for the omitted descriptors are deemed zero). As seen from the figure, the exact values of the coefficients depend on the temporal centrality under consideration. Nevertheless, we can see some trends. The descriptors that show a strong positive correlation with the evader's ability to hide from at least two temporal centralities are the mean of the intercontact time of the evader's contacts $\nu_m$ and the mean of the number of other contacts between two consecutive contacts over node's contacts $\eta_m$. On the other hand, the descriptors that show negative correlation vary greatly between centrality measures. Interestingly, the local clustering coefficient of the evader $\zeta$ has positive regression coefficients values for more locally-oriented centralities, i.e., degree and eigenvector, but negative values for more global centralities, i.e., closeness and betweenness.

\section{Conclusions}
\label{sec:conclusions}

In this article, we analyzed both theoretically and empirically the problem faced by an evader---a member of a temporal social network who wishes to avoid being detected by temporal centrality measures. As part of the theoretical analysis, we defined the decision and the optimization versions of the problem faced by the evader. We proved that the decision version is NP-complete even for very basic temporal centrality measures. As for the optimization version, we presented a $2$-approximation algorithm for the degree temporal centrality while showing that the problem is inapproximable within logarithmic bounds for the closeness and betweenness centrality measures. As part of the empirical analysis, we compared the effectiveness of several hiding heuristics in real-life temporal social networks, finding that removing existing contacts is significantly more effective in avoiding detection by temporal centralities than adding new contacts. Moreover, using regression analysis, we determined that the nodes whose contacts are more distributed throughout the analyzed period are, on average, more successful in obscuring their central position in the network.
More broadly, our study contributes to the literature on hiding from social network analysis tools by considering temporal networks, which are not only more general, but also much more challenging to analyze compared to static networks.

Our findings suggest that while it is computationally infeasible to find an optimal way of hiding from temporal centrality measures, skillful manipulation of one's contacts can still help maintain some semblance of safety even when utilizing heuristic solutions. At the same time, considering the temporal aspects of the network activity offers much more diverse ways of hiding one's importance. Whereas in a static network, a member of the network can only decide who to connect or avoid, in a temporal network, they can also determine the distribution of the contacts, thus creating a much richer strategic space that should be further analyzed in the future research.

\section*{Acknowledgments}

P.H. was supported by JSPS KAKENHI Grant Number JP 21H04595.

\bibliographystyle{abbrv}
\bibliography{bibliography-temporal}

\end{document}